\documentclass[preprint,12pt,3p]{elsarticle}

\usepackage{amsfonts,amsmath,amsthm,amstext,amssymb,url}

 \usepackage{multirow}
 \usepackage{longtable}
 \usepackage{float}

\newtheorem{thm}{Theorem}[section]
\newtheorem{defin}[thm]{Definition}
\newtheorem{lem}[thm]{Lemma}
\newtheorem{cor}[thm]{Corollary}

\newtheorem{examp}[thm]{Example}
\newtheorem{propst}[thm]{Proposition}

\usepackage{multicol}
\usepackage[british]{babel}
\usepackage{amsmath}
\usepackage{color}
\usepackage[font=small,labelfont=bf]{caption}

\def\eqdef{\buildrel{\rm def}\over =}

\def\Q{\mathbb{Q}}
\def\R{\mathbb{R}}

\begin{document}

\everymath{\displaystyle}

%\journal{Computer Aided Geometric Design}

%\definecolor{bg}{rgb}{0.95,0.95,0.95}
\begin{frontmatter}

\title{Computing the intersection of two quadrics \\ through projection and lifting}

\author[label1]{Alexandre Trocado}
\address[label1]{Universidade Aberta, Portugal}
\ead{mail@alexandretrocado.com}

\author[label5]{Laureano Gonzalez-Vega\footnote{Partially supported by the Spanish Ministerio de Ciencia, Innovacci\'on y Universidades under the Project MMTM2017-88796-P.}}
\address[label5]{Universidad de Cantabria, Spain}
\ead{laureano.gonzalez@unican.es}

\begin{abstract} This paper is devoted to presenting a new approach to determine the intersection of two quadrics based on the detailed analysis of its projection in the plane (the so called cutcurve) allowing to perform the corresponding lifting correctly. This approach is based on a new computational characterisation of the singular points of the cutcurve and on how this curve is located with respect to the projection of the considered quadrics (whose boundaries are the so called silhouette curves).
\end{abstract}


\end{frontmatter}

% \linenumbers

\section{Introduction}\label{sec1}
Quadrics are the simplest curved surfaces used in many areas and computing their intersection is a relevant problem. Algorithms dealing with this problem based on floating point arithmetic techniques are sensitive to rounding errors achieving a low running time to the detriment of their correctness. On the other hand, using symbolic methods guarantees the correctness of the results because they are based on exact arithmetic (if the considered quadrics are defined in exact terms) but their performance is typically and significantly lower than using methods based on numerically techniques  (\cite{Berberich:2005:ECE:1064092.1064110,Sendra1999}). 

Levin (\cite{Levin:1976:PAD:360349.360355},\cite{LEVIN197973}) developed a method to parameterize the intersection curve of two quadrics based on the analysis of the pencil generated by them. Wilf \textit{et al.} (\cite{Wilf1993}) improved Levin's ruled-surface parameterization scheme. However, Levin's method often fails to find the intersection curve when it is singular and generates a parameterization that involves the square root of some polynomial (\cite{DUPONT2008168}). Also, when working with floating point numbers, sometimes Levin's method outputs results that are topologically wrong and even fail to produce any parameterization (\cite{DUPONT2008168}). Wang \textit{et al.} (\cite{Wang2002}) reduced the computation of the intersection curve to the analysis of plane cubic curves. Farouki \textit{et al.} (\cite{Farouki1989}) made a complete study of the degenerated cases of quadric intersection by using factorization of multivariate polynomials and Segre characteristic. This method showed the exact parameterization of the intersection curve in many cases.

Later Wang \textit{et al}. (\cite{WANG2003401}) improved Levin's method making it capable of computing geometric and structural information - irreducibility, singularity and the number of connected components. Dupont \textit{et al.} (\cite{DUPONT2008168}, \cite{Dupont2008c}, \cite{Dupont2008c3}) presented a near-optimal algorithm for computing the explicit representation of the intersection of two arbitrary quadrics whose coefficients are rational numbers in the projective space by using the reduction of quadratic forms and producing new results characterising the intersection curve of two quadrics. The performance of this algorithm was analysed in \cite{LAZARD200674}. 

Others have restricted the kind of quadrics to be considered and defined specific routines to each case (\cite{Miller1987}, \cite{Goldman:1991:CAR:112515.112545}, \cite{MILLER199555}, \cite{JOHNSTONE1992179}, \cite{Shene:1994:LDI:195826.197316}) taking advantage of the fact that a geometric approach is typically more stable than the algebraic ones (\cite{DUPONT2008168}). However these approaches are limited to planar intersections and natural quadrics. Mourrain et al. (\cite{Mourrain2005}) proposed an algorithm that reduces the intersection of two quadrics to a dynamic two--dimensional problem.

An alternative way to compute the intersection of two quadric surfaces in the three-dimensional space is based on analysing its projection onto one plane (\cite{Geismann:2001:CCA:378583.378689}). The idea of this method is to reduce the three-dimensional problem to computing the arrangement of three plane algebraic curves defined implicitely. After determining and analysing the projection of the intersection curve onto a plane, the intersection curve can be recovered by determining the lifting of the projection curve. Implementation and theoretical aspects of this approach are also described in \cite{Berberich:2005:ECE:1064092.1064110} and \cite{Schomer2006}, respectively.

In this paper, a new method is presented to determine the intersection curve of two quadrics through projection onto a plane and lifting. In some cases, it will be possible to determine the exact parameterization of the intersection curve (involving radicals if needed) and, in others, the output (topologically correct) will be presented as the lifting of the discretisation of the branches of the projection curve once  its singular points have been fully determined. The way the lifting will be made is the main criteria followed to analyse the cutcurve.

This paper is organized as follows. In Section 2, we briefly review some preliminaries on conics and quadrics. In Section 3 some mathematical tools as resultants and subresultants are briefly presented for sake of completeness. Resultants are used in Section 4 to characterise the projection of the intersection curve (called, in what follows, the cutcurve) by using a bivariate polynomial of degree four, at most. Our approach is based on the analysis of the arrangement of the cutcurve and the silhouette of both quadrics, as in Figure \ref{admiss}, following \cite{Berberich:2005:ECE:1064092.1064110} and \cite{Schomer2006}. Section 5 is devoted to introduce simpler methods to characterise the singular points of the cutcurve as well as its lifting by using the subresultants. Some examples are given in Section 6 where the results of the implementation in Maple are reported and the conclusions are presented in Section 7.

In order to simplify the description of the algorithm and to avoid a long case-by-case analysis we will deal only with quadrics whose defining equations have the shape $z^2+p_1(x,y)z+p_0(x,y)$ and $z^2+q_1(x,y)z+q_0(x,y)$. The general case is very easy to derive from what we describe here.

\begin{figure}[hbt]
\centering
\fbox{\includegraphics[scale=0.4]{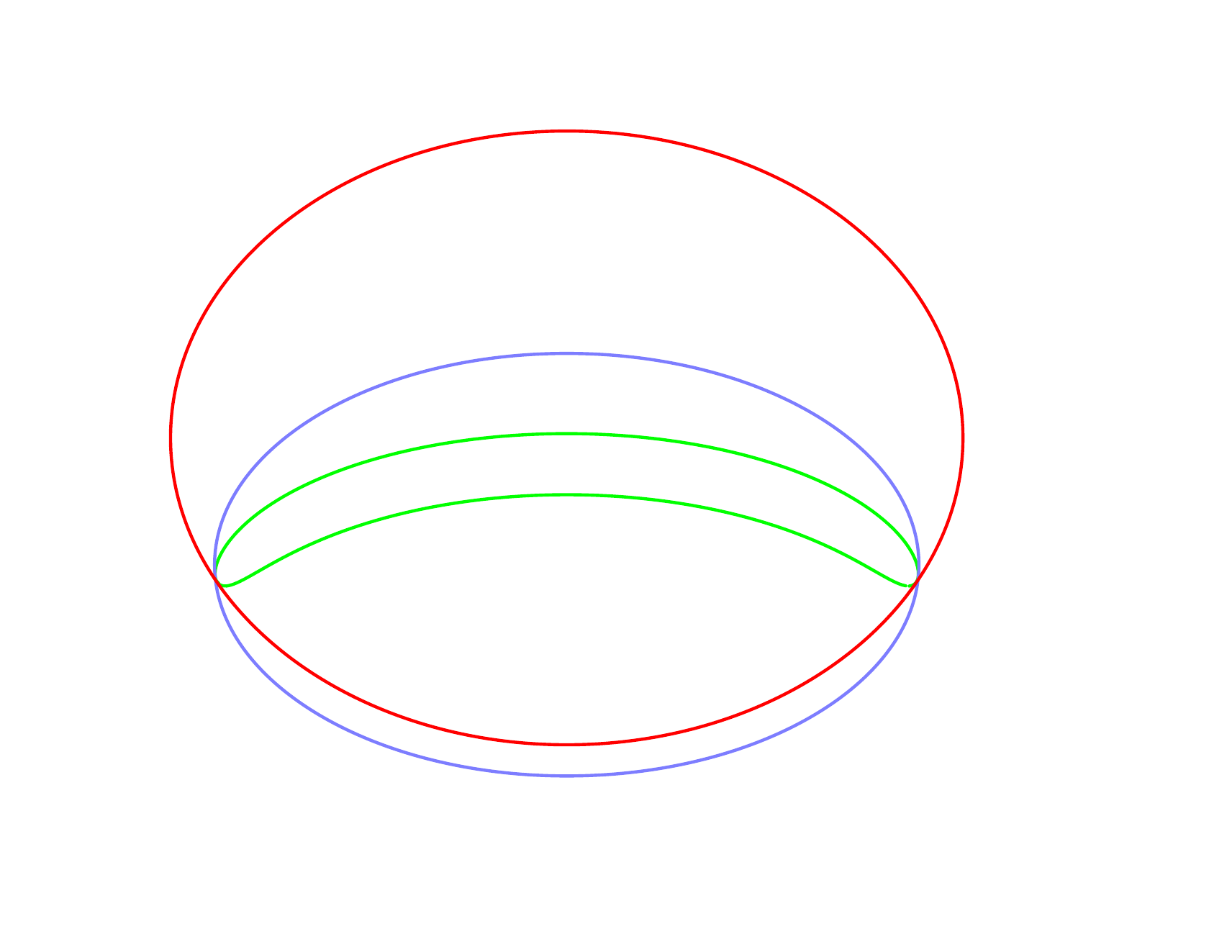}}\hskip 2cm
\fbox{\includegraphics[scale=0.45]{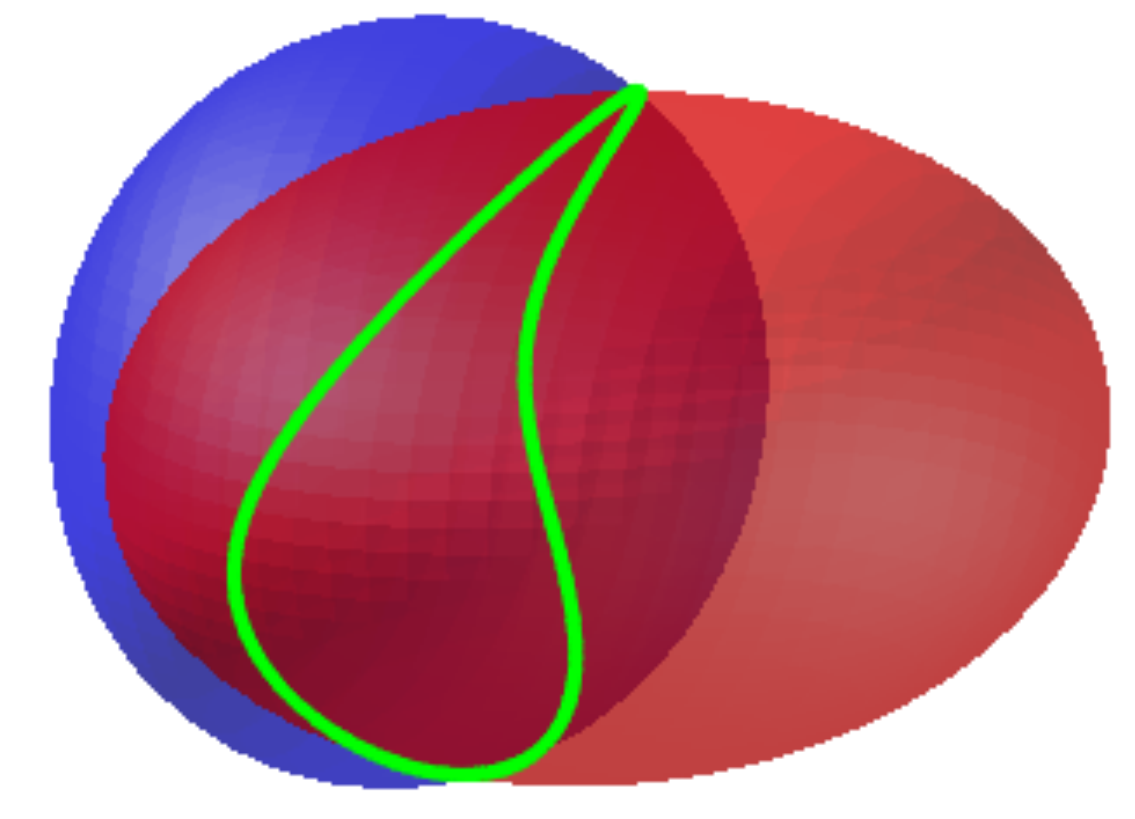}}
\caption{\label{admiss} Left: Silhouette curves of  two quadrics (in red and blue) and the cutcurve (in green) in the plane. Right: the intersection curve and the two quadrics}
\end{figure}

\section{Representing quadrics (and conics)}
This section is devoted to introduce how quadrics will be represented when computing their intersection curve. Since we will project the considered quadrics onto the $xy$ plane and the boundary of this region will be a finite number of conic arcs we introduce here how these regions will be represented and manipulated.

\subsection{Quadrics and conics}
Quadrics are the one of the simplest  surfaces defined by  degree two  polynomials in $x$, $y$ and $z$. The equation of any quadric $\mathcal{A}$ in $\R^3$ can be written as
$$a_{11}x^{2}+a_{22}y^{2}+a_{33}z^{2}
+2a_{12}xy+2a_{13}xz+2a_{23}yz
+2a_{14}x+2a_{24}y+2a_{34}z
+a_{44}=0$$
or in matricial form
$\begin{pmatrix}
   x & y  & z & 1 \\
\end{pmatrix}A\begin{pmatrix}
   x & y & z  & 1 \\
\end{pmatrix}^{T}=0$
where $A$ is the symmetric matrix
$$A=\begin{pmatrix}
   a_{11} & a_{12} & a_{13} & a_{14} \\
   a_{12} & a_{22} & a_{23} & a_{24} \\
   a_{13} & a_{23} & a_{33} & a_{34} \\
   a_{14} & a_{24} & a_{34} & a_{44} \\
\end{pmatrix}.$$

For conics the treatment is analogous (and standard). They will be used later to define the boundary of regions in the plane, typically the region where the projection of the intersection curve of the two considered curves live. In what follows we assume that it is easy to determine the intersection points of two conics and to manipulate the regions of the plane defined by two inequalities involving degree two or one polynomials.

\section{Mathematical tools} \label{sec2}
In this section we will make a brief introduction to resultants and subresultants and how they will be applied later to compute the intersection curve between two quadrics.

\subsection{Resultants}
Resultants and subresultants will be the algebraic tool to use to determine, both, the projection of the intersection curve between the two considered quadrics and its lifting from the plane to the space since they allow a very easy and compact way to characterise the greatest common divisor of two polynomials ($f(x,y,z)$ and $g(x,y,z)$ in our case) when they involve parameters ($x$ and $y$ in our case, since we are going to eliminate $z$).

\begin{defin}\label{subres}
Let  $$P(T)=\sum_{i=0}^m a_{m-i}T^i\qquad \hbox{and} \qquad Q(T)=\sum_{i=0}^n b_{n-i}T^i$$ be two polynomials with coefficients in a field ($\Q$ or $\R$ in our case).  We define the $j$--th subresultant polynomial of $P$ and $Q$ with respect to $T$ in the following way (as in \cite{Li2006a}):
\[\mbox{\bf Sres}_j(P,Q;T)=
\left\vert
\begin{array}{cccccccccc}
a_0 & a_1   &  a_2  &\ldots  &\ldots&a_m    &&&\\
       &\ddots&\ddots &\ddots&         &&\ddots&&\\
       &          &  a_0  & a_1 &a_2 &\ldots  &\ldots&a_m    &\\
       &          &           &       &       &    1     &-T&&\\
       &          &           &       &       &          &\ddots&\ddots&\\
       &          &           &       &       &&&    1     &-T\\
b_0 & b_1   &  b_2  &\ldots  &\ldots&\ldots    &b_n&&\\
       &\ddots&\ddots &\ddots&         &&&\ddots&\\
       &          &  b_0  & b_1 &b_2 &\ldots  &\ldots&\ldots&b_n
\end{array}
\right\vert
\!\!\!\!\!\!\!\
\begin{array}{l}
\left.\begin{array}{c}\\ \\ \\  \\  \end{array}\right\}n-j
\\
\left.\begin{array}{c}\\ \\ \\  \end{array}\right\} j
\\
\left.\begin{array}{c}\\ \\ \\  \\  \end{array}\right\}m-j
\end{array}
\]
and we define the $j$--th subresultant coefficient  of $P$ and $Q$ with respect to $T$,
$\mathrm{\bf sres}_j(P,Q;T)$, as the coefficient of $T^j$ in $\mathrm{\bf Sres}_j(P,Q;T)$.
The resultant of $P$ and $Q$ with respect to $T$ is:
$$\mathrm{\bf Resultant}(P,Q;T)= \mathrm{\bf Sres}_0(P,Q;T)=\mathrm{\bf sres}_0(P,Q;T)\enspace.$$ \end{defin}

There are many different ways of defining and computing subresultants: see \cite{Gonzalez-Vega2009} for a short introduction and for a pointer to several references.

Subresultants allow to characterise easily the degree of the greatest common divisor of two univariate polynomials whose coefficients depend on one or several parameters. Since the resultant of $P$ and $Q$ is equal to
the polynomial ${\bf sres}_{0}(P,Q;T)$,  
${\bf sres}_0(P,Q;T)= 0$ if and only if there exists $T_0$ such that $P(T_0)=0$ and $Q(T_0)=0$. 

More generally, the determinants ${\bf sres}_{j}(P,Q;T)$, which are the
formal leading coefficients of the subresultant sequence for $P$
and $Q$, can be used to compute the greatest common divisor
of $P$ and $Q$ thanks to the following equivalence:
\begin{equation}\label{gcd} {\bf Sres}_i(P,Q;T)=\gcd(P,Q)\Longleftrightarrow
\begin{cases}{\bf sres}_0(P,Q;T)=\ldots={\bf sres}_{i-1}(P,Q;T)=0& \cr
            \hfill{\bf sres}_i(P,Q;T)\neq 0\hfill&
\cr\end{cases}\end{equation}

Let $f$ and $g$ be the two polynomials in $\mathbb{R}[x,y,z]$ 
$$f(x,y,z)= z^2+p_1(x,y)z+p_0(x,y)\qquad 
g(x,y,z)= z^2+q_1(x,y)z+q_0(x,y)$$
($\deg(p_1(x,y)) \le 1$, $\deg(p_0(x,y))\le 2$, $\deg(q_1(x,y)) \le 1$ and $\deg(q_0(x,y))\le 2$) defining the quadrics whose intersection curve is to be computed. Then
the resultant of $f$ and $g$, with respect to $z$, is equal to:
$${\bf S}_0(x,y)\eqdef\mathrm{\bf Resultant}(f,g;z)=\left| {\begin{array}{*{20}{c}}
1&{{p_1(x,y)}}&{{p_0(x,y)}}&0\\
0&1&{{p_1(x,y)}}&{{p_0(x,y)}}\\
1&{{q_1(x,y)}}&{{q_0(x,y)}}&0\\
0&1&{{q_1(x,y)}}&{{q_0(x,y)}}
\end{array}} \right| =$$
$$= {\left( {{p_0(x,y)} - {q_0(x,y)}} \right)^2} - \left( {{p_1(x,y)} - {q_1(x,y)}} \right)\left| {\begin{array}{*{20}{c}}
{{p_0(x,y)}}&{{p_1(x,y)}}\\
{{q_0(x,y)}}&{{q_1(x,y)}}
\end{array}} \right| .$$
The degree of ${\bf S}_0(x,y)$ is at most four.
The first subresultant of $f$ and $g$, with respect to $z$, is equal to:
$${\bf S}_1(x,y;z)\eqdef\mbox{\bf Sres}_1(f,g;z)=(q_1(x,y)-p_1(x,y))z+(q_0(x,y)-p_0(x,y))=g(x,y,z)-f(x,y,z).$$

Computing the intersection of the two quadrics defined by $f$ and $g$ is equivalent to solving in $\R^3$ the polynomial system of equations $$f(x,y,z)=0, \qquad g(x,y,z)=0.$$ The solution set to be computed, when non empty, may include curves and isolated points. We will use that the above polynomial system of equations, under some conditions, is equivalent to 
$${\bf S}_0(x,y)=0,\qquad (q_1(x,y)-p_1(x,y))z+(q_0(x,y)-p_0(x,y))=0.$$
Analyzing ${\bf S}_0(x,y)=0$ in $\R^2$ will be called the  projection step and moving the information obtained in $\R^2$ to $\R^3$ will be called the lifting step. We follow here the terminology used when computing the cylindrical algebraic decomposition of a finite set of multivariate polynomials (see \cite{Basu2006}, for example)

\section{Projecting the intersection curve}
\noindent In this section we will characterise the projection of the intersection curve of  two quadrics onto the $(x,y)$--plane. The usual way of dealing with projections of algebraic sets involves tools coming from Elimination Theory. Since we project from $\R^3$ onto $\R^2$ such a description will involve  polynomial inequalities.

Let $f$ and $g$ be the two polynomials in $\mathbb{R}[x,y,z]$ 
\begin{equation}\label{22quadrics}
f(x,y,z)=z^2+p_1(x,y)z+p_0(x,y) \qquad
g(x,y,z)=z^2+q_1(x,y)z+q_0(x,y)
\end{equation}
with $\deg(p_1) \le 1$, $\deg(p_0)\le 2$, $\deg(q_1) \le 1$ and $\deg(q_0)\le 2$. 

Let 
$\Delta _{{\cal E}_1}(x,y) = p_1(x,y)^2-4p_0(x,y)$ and 
$\Delta _{{\cal E}_2}(x,y) = q_1(x,y)^2-4q_0(x,y)$ be the discriminants of $f(x,y,z)$ and $g(x,y,z)$ (respectively) with respect to $z$.

Let ${\cal E}_1$ and ${\cal E}_2$ be the corresponding quadrics:
$${\cal E}_1:\left\{ {\left( {x,y,z} \right) \in {\mathbb{R}^3}:f\left( {x,y,z} \right)=0} \right\}\qquad
{\cal E}_2:\left\{ {\left( {x,y,z} \right) \in {\mathbb{R}^3}:g\left( {x,y,z} \right)=0} \right\},
$$
and $\Pi$ be the projection :
\[\begin{array}{*{20}{c}}
{\Pi :}&{{\mathbb{R}^3}}& \to &{{\mathbb{R}^2}}\\
{}&{\left( {x,y,z} \right)}& \mapsto &{\left( {x,y} \right)}
\end{array}\]
Next (easy to prove) theorem characterises the set $\pi \left( {\cal E}_1\cap {\cal E}_2 \right)$, the projection of the intersection curve of ${\cal E}_1$ and ${\cal E}_2$.

\begin{thm}\label{projectioncurve}
 \[
 \Pi \left( {\cal E}_1\cap {\cal E}_2 \right)
 =  
 \left\{ (x,y) \in \mathbb{R}^2: 
{\bf S}_0(x,y) = 0,\Delta _{{\cal E}_1}(x,y)\geq 0,\Delta _{{\cal E}_2}(x,y)\geq 0
 \right\}.
 \]
\end{thm}

\begin{proof}
If $(a,b)\in\Pi \left( {\cal E}_1 \cap {\cal E}_2 \right)$ then there exists $c\in\mathbb{R}$ such that $(a,b,c)\in {\cal E}_1 \cap {\cal E}_2$. Thus $f(a,b,c)=0$ and $g(a,b,c)=0$, $f(a,b,z)$ and $g(a,b,z)$ have a common root ($c\in\mathbb{R}$) and we conclude that ${\bf S}_0(a,b)=0$. As $c$ is a real root of $f(a,b,z)$ and $g(a,b,z)$ we also have 
$\Delta _{{\cal E}_1}(a,b)\geq 0$ and $\Delta _{{\cal E}_2}(a,b)\geq 0$.

On the other hand, if $(a,b)\in\mathbb{R}^2$ verifies ${\bf S}_0(a,b)=0$ then $f(a,b,z)$ and $g(a,b,z)$ have a common root $c \in \mathbb{C}$: $f(a,b,c)=0$ and $g(a,b,c)=0$ (according to (\ref{gcd})). However, if $\Delta _{{\cal E}_1}(a,b)\ge0$ and $\Delta _{{\cal E}_2}(a,b)\ge0$ then $c$ must be a real solution of $f(a,b,z)=0$ and $g(a,b,z)=0$, concluding that $(a,b)\in \Pi \left({{{\cal E}_1} \cap {{\cal E}_2}}\right)$.
\end{proof}

Previous theorem gives the precise description of the projection of the intersection curve of two quadrics when their defining equations have the structure introduced in (\ref{22quadrics}). It corresponds to the part of the curve 
$$\{(x,y)\in\R^2\colon{\bf S}_0(x,y)=0\}$$ 
inside the region 
$${\cal A}_{{\cal E}_1,{\cal E}_2}\eqdef\{(x,y)\in\R^2\colon\Delta _{{\cal E}_1}(x,y)\geq 0, \Delta _{{\cal E}_2}(x,y)\geq 0\}.$$
As expected this set is a semialgebraic set in $\R^2$ since, according to Tarski Principle (see \cite{Basu2006}), the projection of any semialgebraic set is a semialgebraic set too.
Thus, the region where  the projection of the intersection curve of the two considered quadrics lives is bounded by a finite set of conic arcs since any $\Delta _{{\cal E}_i}(x,y)$ is a polynomial in $\R[x,y]$ of total degree equal to $2$.

In \cite{Berberich:2005:ECE:1064092.1064110}, the curve in $\R^2$ defined by ${\bf S}_0(x,y)=0$ is called the cutcurve of ${\cal E}_1$ and ${\cal E}_2$ and the curve in $\R^2$ defined by {$\Delta _{{\cal E}_i}(x,y)=0$} the silhouette of ${\cal E}_i$. We modify slightly the cutcurve definition to make it more suitable for our purposes.

\begin{defin}
Let ${\cal E}_1$ and ${\cal E}_2$ be two quadrics in $\R^3$ defined by $f(x,y,z)=0$ and $g(x,y,z)=0$ respectively. The cutcurve of ${\cal E}_1$ and ${\cal E}_2$ is the set 
$$\left\{ (x,y) \in \mathbb{R}^2: 
 {\bf S}_0(x,y) = 0, \Delta _{{\cal E}_1}(x,y)\geq 0, \Delta _{{\cal E}_2}(x,y)\geq 0,
 \right\}$$
\end{defin}

According to Theorem \ref{projectioncurve} the cutcurve of ${\cal E}_1$ and ${\cal E}_2$ is equal to the projection of ${\cal E}_1\cap{\cal E}_2$ onto the $xy$ plane, $\Pi({\cal E}_1\cap{\cal E}_2)$. The cutcurve of ${\cal E}_1$ and ${\cal E}_2$ can be a curve, part of a curve (i.e. a semialgebraic set) or even a single point, but always a semialgebraic set.

\begin{examp}\label{curvareal}{\rm
Let $f$ and $g$ be the polynomials 
$$f(x,y,z)={z^2}+{x^2} + {y^2} - 7 \qquad g(x,y,z)={z^2}-{x^2} + xy + 2x - {y^2}$$
defining the two quadrics ${\cal G}_1$ and ${\cal G}_2$ whose intersection curve is to be computed. In this case we have:
\[
{\bf S}_0(x,y)=(-2{x}^{2}+xy-2{y}^{2}+2x+7 ) ^{2},
\]
\[{\Delta _{{{\cal G} _1}}} (x,y)=  -4{x}^{2}-4{y}^{2}+28,\qquad
{\Delta _{{{\cal G}_2}}}(x,y)=  4{x}^{2}-4xy+4{y}^{2}-8x \]
and
$$
{\cal A}_{{\cal G}_1,{\cal G}_2}=\{(x,y)\in\R^2\colon -{x}^{2}-{y}^{2}+7\geq 0,\; {x}^{2}-xy+{y}^{2}-2x\geq 0 \}.$$
In this case the curve in $\R^2$ defined by ${\bf S}_0(x,y)=0$ is not contained completely in 
${\cal A}_{{\cal G}_1,{\cal G}_2}$: the projection of ${\cal G}_1\cap{\cal G}_2$ is equal to the portion of the ellipse $-2{x}^{2}+xy-2{y}^{2}+2x+7=0$ inside the circle ${x}^{2}+{y}^{2}\leq 7$ (see Figure \ref{test1}).}
\end{examp}

\begin{figure}[hbt]
\centering
\includegraphics[scale=0.15]{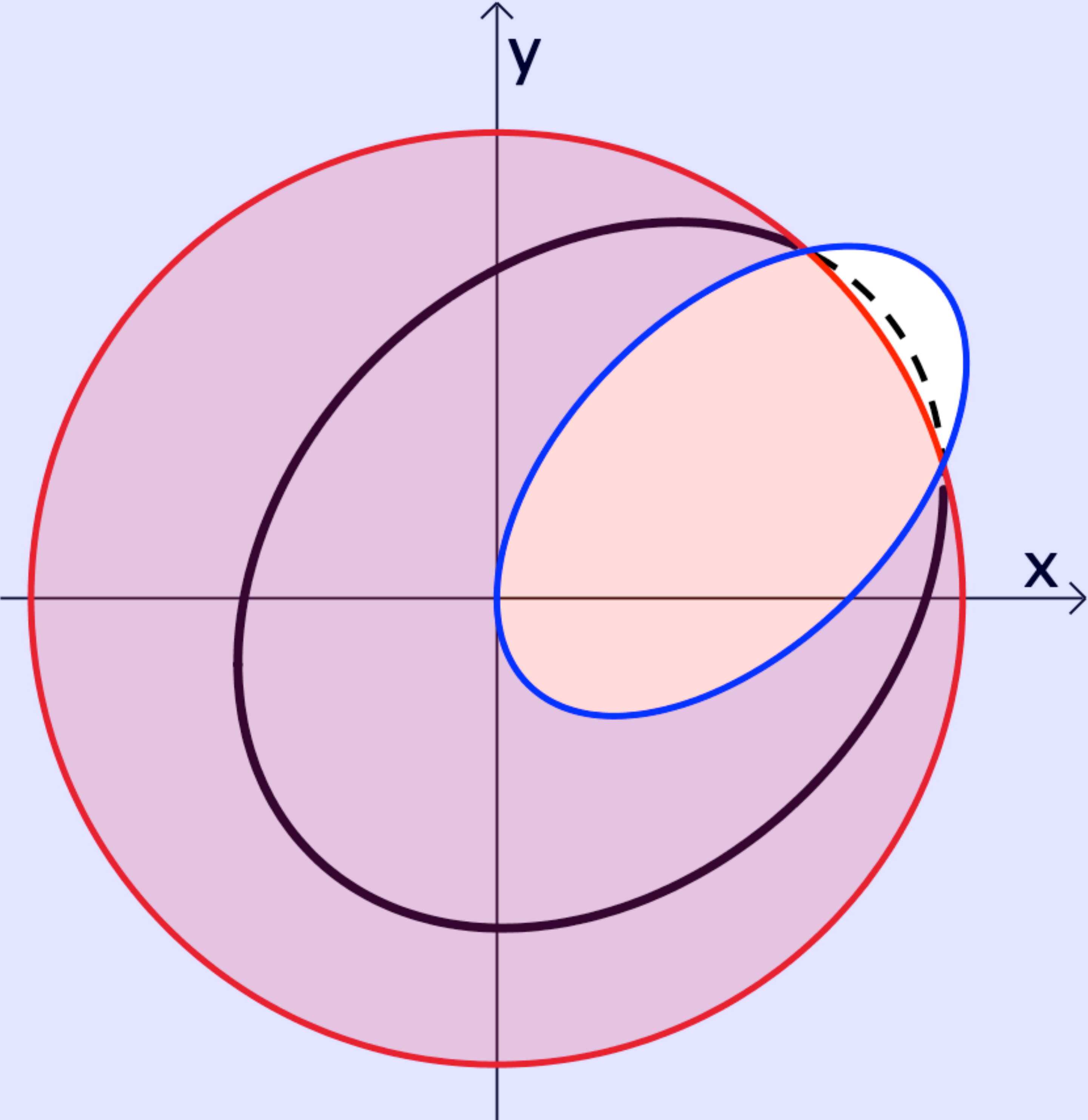}\hskip 2cm
\includegraphics[scale=0.25]{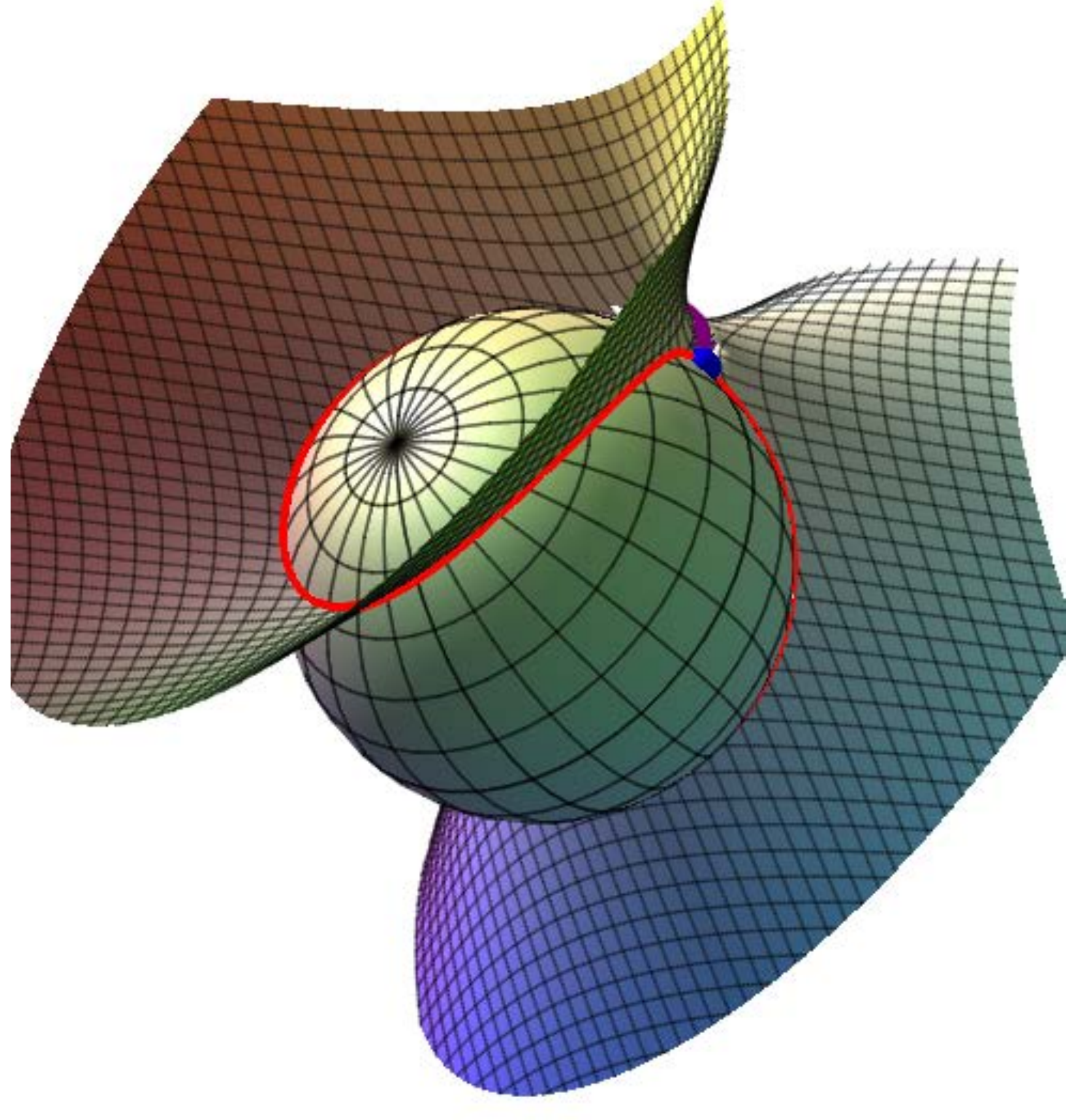}
\caption{\label{test1} Example where the cutcurve (black) is not the whole curve ${\bf S}_0(x,y)=0$.}
\end{figure}

We will pay also special attention to the intersection points between the cutcurve and the silhouette curves. These points will include those points where the cutcurve ``stops" and where the cutcurve ``starts" again (unless both curves are tangent). Figure \ref{test1} shows one concrete example of this situation.

For quadrics whose defining equations have a different structure than the shown in (\ref{22quadrics}), similar formulae can be derived for the cutcurve, i.e. for the projection of their intersection. For example, if
\begin{equation}\label{21quadrics}
f(x,y,z)=z^2+p_1(x,y)z+p_0(x,y) \qquad
g(x,y,z)=q_1(x,y)z+q_0(x,y)
\end{equation}
then $\Pi \left( {\cal E}_1\cap {\cal E}_2 \right)$ is given by
$$
\left\{ (x,y) \in \mathbb{R}^2: 
 \begin{array}{l}
 {\bf S}_0(x,y) = 0\\ p_1(x,y)^2-4p_0(x,y)\geq 0\\ q_1(x,y)\neq 0
\end{array}
 \right\}
 \bigcup
\left\{ (x,y) \in \mathbb{R}^2: 
 \begin{array}{l}
 {\bf S}_0(x,y) = 0\\ p_1(x,y)^2-4p_0(x,y)\geq 0\\ q_1(x,y)=0,\; q_0(x,y)=0
\end{array}
 \right\}.
$$

\section{Lifting to $\R^3$ the cutcurve in $\R^2$}\label{lifting}
In this section we  study the lifting to $\R^3$ of the cutcurve of the two quadrics whose intersection is to be computed. We will pay special attention to the singular points of the cutcurve since they are the points where more complicated situations we must deal with when lifting the cutcurve of ${\cal E}_1$ and ${\cal E}_2$ to ${\cal E}_1\cap{\cal E}_2$. This means that, first, we must be able of isolating them in order to achieve its lifting in the easiest and most efficient possible way. 

\subsection{Determining the singular points of the cutcurve of ${\cal E}_1$ and ${\cal E}_2$}\label{lifting1}
Let $f$ and $g$ be two polynomials in $\mathbb{R}[x,y,z]$ defined by: 
$$
f(x,y,z)= z^2+p_1(x,y)z+p_0(x,y)\qquad
g(x,y,z)=z^2+q_1(x,y)z+q_0(x,y)
$$
with $\deg(p_1) \leq 1$, $\deg(p_0)\leq 2$, $ \deg(q_1) \leq 1$, $\deg(q_0)\leq 2$. We restrict our attention to this case because this is the most complicated situation we must deal with: for those quadrics whose equations have a different (and simpler) structure, the singular points of the cutcurve are easier to compute since their lifting will be given automatically by one of the two equations (being of degree in $z$ smaller than or equal to $1$).

Let ${\cal E}_1$ and ${\cal E}_2$ be the corresponding quadrics defined by $f$ and $g$ 
and ${\bf S}_0(x,y) = 0$ the implicit equation of its cutcurve.
Next two theorems will help to determine easily the singular points of the cutcurve. The first one shows that those points in ${\bf S}_0(x_0,y_0) = 0$ and in the line $p_1(x,y)=q_1(x,y)$ are always singular points of the cutcurve.

In what follows we will use very often the following (easy to prove) lemma. We thank the reviewer for pointing out the importance of this fact that we were using without isolating it properly.

\begin{lem}\label{lemma}
Let $(x_0,y_0)$ be such that ${\bf S}_0(x_0,y_0)=0$, $z_0$ such that:
$$f(x_0,y_0,z_0)=z_0^2+p_1(x_0,y_0)z_0+p_0(x_0,y_0)=0\quad g(x_0,y_0,z_0)=z_0^2+q_1(x_0,y_0)z_0+q_0(x_0,y_0)=0$$
and ${p_1}\left( {{x_0},{y_0}} \right) = {q_1}\left( {{x_0},{y_0}} \right)$. Then ${p_0}\left( {{x_0},{y_0}} \right) = {q_0}\left( {{x_0},{y_0}} \right)$.
\end{lem} 

Next lemma gives a very convenient description of the partial derivatives of  ${\bf S}(x,y)$.

\begin{lem}\label{lemma2}
If ${\bf S}(x_0,y_0)=0$ then there exists $z_0$ such that $f(x_0,y_0,z_0)=0$, $g(x_0,y_0,z_0)=0$ $$\frac{{\partial {{\bf S}_0}}}{{\partial x}}\left( {{x_0},{y_0}} \right)=
\left| {\begin{array}{*{20}{c}}
1&{{p_1}\left( {{x_0},{y_0}} \right)}&{{p_0}\left( {{x_0},{y_0}} \right)}&{{z_0}{f_x}\left( {{x_0},{y_0},{z_0}} \right)}\\
0&1&{{p_1}\left( {{x_0},{y_0}} \right)}&{{f_x}\left( {{x_0},{y_0},{z_0}} \right)}\\
1&{{q_1}\left( {{x_0},{y_0}} \right)}&{{q_0}\left( {{x_0},{y_0}} \right)}&{{z_0}{g_x}\left( {{x_0},{y_0},{z_0}} \right)}\\
0&1&{{q_1}\left( {{x_0},{y_0}} \right)}&{{g_x}\left( {{x_0},{y_0},{z_0}} \right)}
\end{array}} \right|,$$
$$\frac{{\partial {{\bf S}_0}}}{{\partial y}}\left( {{x_0},{y_0}} \right)=
\left| {\begin{array}{*{20}{c}}
1&{{p_1}\left( {{x_0},{y_0}} \right)}&{{p_0}\left( {{x_0},{y_0}} \right)}&{{z_0}{f_y}\left( {{x_0},{y_0},{z_0}} \right)}\\
0&1&{{p_1}\left( {{x_0},{y_0}} \right)}&{{f_y}\left( {{x_0},{y_0},{z_0}} \right)}\\
1&{{q_1}\left( {{x_0},{y_0}} \right)}&{{q_0}\left( {{x_0},{y_0}} \right)}&{{z_0}{g_y}\left( {{x_0},{y_0},{z_0}} \right)}\\
0&1&{{q_1}\left( {{x_0},{y_0}} \right)}&{{g_y}\left( {{x_0},{y_0},{z_0}} \right)}
\end{array}} \right|.$$
\end{lem}

\begin{proof}
The existence of $z_0$ comes from the fact that ${\bf S}_0(x,y)$ is the resultant of $f(x,y,z)$ and $g(x,y,z)$ with respect to $z$. By performing elementaty operations on the columns of the matrix giving ${\bf S}_0(x,y)$ we obtain:
$${\bf S}_0(x,y)=\mathrm{\bf Resultant}(f,g;z)=\left| {\begin{array}{*{20}{c}}
1&{{p_1}}&{{p_0}}&0\\
0&1&{{p_1}}&{{p_0}}\\
1&{{q_1}}&{{q_0}}&0\\
0&1&{{q_1}}&{{q_0}}
\end{array}} \right| = \left| {\begin{array}{*{20}{c}}
1&{{p_1}}&{{p_0}}&{zf}\\
0&1&{{p_1}}&f\\
1&{{q_1}}&{{q_0}}&{zg}\\
0&1&{{q_1}}&g
\end{array}} \right|.$$ And applying the rule for calculating the derivative of a determinant we have:
$$
\frac{{\partial {{\bf S}_0}}}{{\partial x}}=\left| {\begin{array}{*{20}{c}}
0&{{p_1}}&{{p_0}}&{zf}\\
0&1&{{p_1}}&f\\
0&{{q_1}}&{{q_0}}&{zg}\\
0&1&{{q_1}}&g
\end{array}} \right| + \left| {\begin{array}{*{20}{c}}
1&{{p_{1x}}}&{{p_0}}&{zf}\\
0&0&{{p_1}}&f\\
0&{{q_{1x}}}&{{q_0}}&{zg}\\
0&0&{{q_1}}&g
\end{array}} \right| + \left| {\begin{array}{*{20}{c}}
1&{{p_1}}&{{p_{0x}}}&{zf}\\
0&1&{{p_{1x}}}&f\\
1&{{q_1}}&{{q_{0x}}}&{zg}\\
0&1&{{q_{1x}}}&g
\end{array}} \right| + \left| {\begin{array}{*{20}{c}}
1&{{p_1}}&{{p_0}}&{z{f_x}}\\
0&1&{{p_1}}&{{f_x}}\\
1&{{q_1}}&{{q_0}}&{z{g_x}}\\
0&1&{{q_1}}&{{g_x}}
\end{array}} \right| .
$$
Replacing $(x,y,z)$ by $(x_0,y_0,z_0)$ in this equation produces:
$$\frac{{\partial {{\bf S}_0}}}{{\partial x}}\left( {{x_0},{y_0}} \right)=
\left| {\begin{array}{*{20}{c}}
1&{{p_1}\left( {{x_0},{y_0}} \right)}&{{p_0}\left( {{x_0},{y_0}} \right)}&{{z_0}{f_x}\left( {{x_0},{y_0},{z_0}} \right)}\\
0&1&{{p_1}\left( {{x_0},{y_0}} \right)}&{{f_x}\left( {{x_0},{y_0},{z_0}} \right)}\\
1&{{q_1}\left( {{x_0},{y_0}} \right)}&{{q_0}\left( {{x_0},{y_0}} \right)}&{{z_0}{g_x}\left( {{x_0},{y_0},{z_0}} \right)}\\
0&1&{{q_1}\left( {{x_0},{y_0}} \right)}&{{g_x}\left( {{x_0},{y_0},{z_0}} \right)}
\end{array}} \right|.$$
Replacing $x$ by $y$ in the previous argument we prove also the remaining equality.
\end{proof}

\begin{thm}\label{singular1}
If  ${\bf S}_0(x_0,y_0) = 0$ and $p_1(x_0,y_0)=q_1(x_0,y_0)$ then $$\frac{{\partial {{\bf S}_0}}}{{\partial x}}\left( {{x_0},{y_0}} \right) =0\qquad  \frac{{\partial {{\bf S}_0}}}{{\partial y}}\left( {{x_0},{y_0}} \right) = 0.$$
\end{thm}
\begin{proof}
Since $p_1(x_0,y_0)=q_1(x_0,y_0)$, by using Lemma \ref{lemma} and Lemma \ref{lemma2}, we get
$$\frac{{\partial {{\bf S}_0}}}{{\partial x}}\left( {{x_0},{y_0}} \right)=\left| {\begin{array}{*{20}{c}}
1&{{p_1}\left( {{x_0},{y_0}} \right)}&{{p_0}\left( {{x_0},{y_0}} \right)}&{{z_0}{f_x}\left( {{x_0},{y_0},{z_0}} \right)}\\
0&1&{{p_1}\left( {{x_0},{y_0}} \right)}&{{f_x}\left( {{x_0},{y_0},{z_0}} \right)}\\
1&{{p_1}\left( {{x_0},{y_0}} \right)}&{{p_0}\left( {{x_0},{y_0}} \right)}&{{z_0}{g_x}\left( {{x_0},{y_0},{z_0}} \right)}\\
0&1&{{p_1}\left( {{x_0},{y_0}} \right)}&{{g_x}\left( {{x_0},{y_0},{z_0}} \right)}
\end{array}} \right|=$$
$$=\left| {\begin{array}{*{20}{c}}
1&{{p_1}\left( {{x_0},{y_0}} \right)}&{{p_0}\left( {{x_0},{y_0}} \right)}&{{z_0}{f_x}\left( {{x_0},{y_0},{z_0}} \right)}\\
0&1&{{p_1}\left( {{x_0},{y_0}} \right)}&{{f_x}\left( {{x_0},{y_0},{z_0}} \right)}\\
0&0&0&{{z_0}{g_x}\left( {{x_0},{y_0},{z_0}} \right)}\\
0&0&0&{{g_x}\left( {{x_0},{y_0},{z_0}} \right)}
\end{array}} \right|=0.
$$
Replacing $x$ by $y$ in the previous argument we prove also that $\frac{{\partial {{\bf S}_0}}}{{\partial y}}
\left( {{x_0},{y_0}} \right) = 0$.
\end{proof}

Theorem \ref{singular1} and Lemma \ref{lemma} imply that  the singular points of the cutcurve in the line $p_1(x,y)=q_1(x,y)$ are also in the conic $p_0(x,y)=q_0(x,y)$.

Next theorem and corollary, together with what we have just proven, allows to conclude that the singular points of the cutcurve come from two different sources:
\begin{itemize}
\item either they are in the line $p_1(x,y)=q_1(x,y)$ (and in the conic $p_0(x,y)=q_0(x,y)$), or
\item they are not in the line $p_1(x,y)=q_1(x,y)$ and they are the projection of  a tangential intersection point between  ${\cal E}_1$ and ${\cal E}_2$.
\end{itemize}
Note that not all tangential intersection points between  ${\cal E}_1$ and ${\cal E}_2$ come from the second option. 

\begin{thm}\label{singular2}
If $(x_0,y_0)$ is a singular point of the cutcurve and $p_1(x_0,y_0)\ne q_1(x_0,y_0)$ then 
\[\frac{{\partial f}}{{\partial x}}\left( {{x_0},{y_0},{z_0}} \right)\frac{{\partial g}}{{\partial z}}\left( {{x_0},{y_0},{z_0}} \right) - \frac{{\partial f}}{{\partial z}}\left( {{x_0},{y_0},{z_0}} \right)\frac{{\partial g}}{{\partial x}}\left( {{x_0},{y_0},{z_0}} \right) = 0\]
\[\frac{{\partial f}}{{\partial y}}\left( {{x_0},{y_0},{z_0}} \right)\frac{{\partial g}}{{\partial z}}\left( {{x_0},{y_0},{z_0}} \right) - \frac{{\partial f}}{{\partial z}}\left( {{x_0},{y_0},{z_0}} \right)\frac{{\partial g}}{{\partial y}}\left( {{x_0},{y_0},{z_0}} \right) = 0\]
where
\[{z_0} =  - \frac{{{q_0}\left( {{x_0},{y_0}} \right) - {p_0}\left( {{x_0},{y_0}} \right)}}{{{q_1}\left( {{x_0},{y_0}} \right) - {p_1}\left( {{x_0},{y_0}} \right)}} .\]
\end{thm}
\begin{proof}
It is enough to check that:
$$\frac{{\partial {{\bf S}_0}}}{{\partial x}}\left( {{x_0},{y_0}} \right) = 
\left( {{q_1}\left( {{x_0},{y_0}} \right) - {p_1}\left( {{x_0},{y_0}} \right)} \right)\cdot
{\left| {\begin{array}{cc}
f_x(x_0,y_0,z_0)&f_z(x_0,y_0,z_0)\\
g_x(x_0,y_0,z_0)&g_z(x_0,y_0,z_0)
\end{array}} \right|}$$
and 
$$\frac{{\partial {{\bf S}_0}}}{{\partial y}}\left( {{x_0},{y_0}} \right) = 
\left( {{q_1}\left( {{x_0},{y_0}} \right) - {p_1}\left( {{x_0},{y_0}} \right)} \right)\cdot
{\left| {\begin{array}{cc}
f_y(x_0,y_0,z_0)&f_z(x_0,y_0,z_0)\\
g_y(x_0,y_0,z_0)&g_z(x_0,y_0,z_0)
\end{array}} \right|}.$$
We only prove the first equality. The proof of the second one is the same but replacing $x$ by $y$.

Let $(x_0,y_0)$ be such that ${\bf S}_0(x_0,y_0)=0$ and $z_0$ such that:
$$f(x_0,y_0,z_0)=z_0^2+p_1(x_0,y_0)z_0+p_0(x_0,y_0)\qquad g(x_0,y_0,z_0)=z_0^2+q_1(x_0,y_0)z_0+q_0(x_0,y_0)=0$$
By using Lemma \ref{lemma2} we have:
$$\frac{{\partial {{\bf S}_0}}}{{\partial x}}\left( {{x_0},{y_0}} \right)= \left| {\begin{array}{*{20}{c}}
1&{{p_1}\left( {{x_0},{y_0}} \right)}&{{p_0}\left( {{x_0},{y_0}} \right)}&{{z_0}{f_x}\left( {{x_0},{y_0},{z_0}} \right)}\\
0&1&{{p_1}\left( {{x_0},{y_0}} \right)}&{{f_x}\left( {{x_0},{y_0},{z_0}} \right)}\\1&{{q_1}\left( {{x_0},{y_0}} \right)}&{{q_0}\left( {{x_0},{y_0}} \right)}&{{z_0}{g_x}\left( {{x_0},{y_0},{z_0}} \right)}\\
0&1&{{q_1}\left( {{x_0},{y_0}} \right)}&{{g_x}\left( {{x_0},{y_0},{z_0}} \right)}
\end{array}} \right|=\left| {\begin{array}{*{20}{c}}
1&{{p_1}}&{{p_0}}&{{z_0}{f_x}}\\
0&1&{{p_1}}&{{f_x}}\\
1&{{p_1}}&{{q_0}}&{{z_0}{g_x}}\\
0&1&{{p_1}}&{{g_x}}
\end{array}} \right|.$$
Then 
$$\frac{{\partial {{\bf S}_0}}}{{\partial x}}\left( {{x_0},{y_0}} \right)=\left| {\begin{array}{*{20}{c}}
1&{{p_1}}&{{p_0}}&{{z_0}{f_x}}\\
0&1&{{p_1}}&{{f_x}}\\
0&{{q_1} - {p_1}}&{{q_0} - {p_0}}&{{z_0}\left( {{g_x} - {f_x}} \right)}\\
0&0&{{q_1} - {p_1}}&{{g_x} - {f_x}}
\end{array}} \right| = \left| {\begin{array}{*{20}{c}}
1&{{p_1}}&{{f_x}}\\
{{q_1} - {p_1}}&{{q_0} - {p_0}}&{{z_0}\left( {{g_x} - {f_x}} \right)}\\
0&{{q_1} - {p_1}}&{{g_x} - {f_x}}
\end{array}} \right| =$$
$$= \left| {\begin{array}{*{20}{c}}
{{q_0} - {p_0}}&{{z_0}\left( {{g_x} - {f_x}} \right)}\\
{{q_1} - {p_1}}&{{g_x} - {f_x}}
\end{array}} \right| - \left( {{q_1} - {p_1}} \right)\left| {\begin{array}{*{20}{c}}
{{p_1}}&{{f_x}}\\
{{q_1} - {p_1}}&{{g_x} - {f_x}}
\end{array}} \right|.$$
Since ${p_1}\left( {{x_0},{y_0}} \right) \ne {q_1}\left( {{x_0},{y_0}} \right)$ then 
$$- {z_0}\left( {{q_1}(x_0,y_0) - {p_1}(x_0,y_0)} \right) = {q_0}(x_0,y_0) - {p_0}(x_0,y_0)$$
and
\begin{eqnarray*}
\frac{{\partial {{\bf S}_0}}}{{\partial x}}\left( {{x_0},{y_0}} \right)& = &\left| {\begin{array}{*{20}{c}}
{ - {z_0}\left( {{q_1} - {p_1}} \right)}&{{z_0}\left( {{g_x} - {f_x}} \right)}\\
{{q_1} - {p_1}}&{{g_x} - {f_x}}
\end{array}} \right| - \left( {{q_1} - {p_1}} \right)\left| {\begin{array}{*{20}{c}}
{{p_1}}&{{f_x}}\\
{{q_1}}&{{g_x}}
\end{array}} \right| = \\
& = &\left| {\begin{array}{*{20}{c}}
{ - {z_0}}&{{z_0}}\\
1&1
\end{array}} \right|\left( {{q_1} - {p_1}} \right)\left( {{g_x} - {f_x}} \right) - \left( {{q_1} - {p_1}} \right)\left| {\begin{array}{*{20}{c}}
{{p_1}}&{{f_x}}\\
{{q_1}}&{{g_x}}
\end{array}} \right|=\\ 
&=& \left( {{q_1} - {p_1}} \right)\left( { - 2{z_0}\left( {{g_x} - {f_x}} \right) - \left| {\begin{array}{*{20}{c}}
{{p_1}}&{{f_x}}\\
{{q_1}}&{{g_x}}
\end{array}} \right|} \right).
\end{eqnarray*}
Since ${f_z}\left( {{x_0},{y_0},{z_0}} \right) = 2{z_0} + {p_1}\left( {{x_0},{y_0}} \right)$ and ${g_z}\left( {{x_0},{y_0},{z_0}} \right) = 2{z_0} + {q_1}\left( {{x_0},{y_0}} \right)$ we have:
$${p_1} = {f_z} - 2{z_0} \qquad {q_1} = {g_z} - 2{z_0}$$
and
\begin{eqnarray*}
\frac{{\partial {{\bf S}_0}}}{{\partial x}}\left( {{x_0},{y_0}} \right)&=&\left( {{q_1} - {p_1}} \right)\left({ - 2{z_0}\left( {{g_x} - {f_x}} \right) - \left| {\begin{array}{*{20}{c}}
{{f_z} - 2{z_0}}&{{f_x}}\\
{{g_z} - 2{z_0}}&{{g_x}}
\end{array}} \right|} \right) = \\
&=& \left( {{q_1} - {p_1}} \right)\left( { - 2{z_0}{g_x} + 2{z_0}{f_x} - {g_x}\left( {{f_z} - 2{z_0}} \right) + {f_x}\left( {{g_z} - 2{z_0}} \right)} \right) =\\
&=&\left( {{q_1} - {p_1}} \right)\left( {{f_x}{g_z} - {f_z}{g_x}} \right).
\end{eqnarray*}
Finally we have 
$$\frac{{\partial {{\bf S}_0}}}{{\partial x}}\left( {{x_0},{y_0}} \right) = \left( {{q_1}\left( {{x_0},{y_0}} \right) - {p_1}\left( {{x_0},{y_0}} \right)} \right)\cdot\left| {\begin{array}{*{20}{c}}
{{f_x}\left( {{x_0},{y_0},{z_0}} \right)}&{{f_z}\left( {{x_0},{y_0},{z_0}} \right)}\\
{{g_x}\left( {{x_0},{y_0},{z_0}} \right)}&{{g_z}\left( {{x_0},{y_0},{z_0}} \right)}
\end{array}} \right| .$$
Since $(x_0,y_0)$ is a singular point of the curve and ${p_1}\left( {{x_0},{y_0}} \right) \ne {q_1}\left( {{x_0},{y_0}} \right)$ we conclude that 
$$\left| {\begin{array}{*{20}{c}}
{{f_x}\left( {{x_0},{y_0},{z_0}} \right)}&{{f_z}\left( {{x_0},{y_0},{z_0}} \right)}\\
{{g_x}\left( {{x_0},{y_0},{z_0}} \right)}&{{g_z}\left( {{x_0},{y_0},{z_0}} \right)}
\end{array}} \right|=0$$ as desired.
\end{proof}

\begin{cor}\label{cortangent}
If $(x_0,y_0)$ is a singular point of the cutcurve, $p_1(x_0,y_0)\ne q_1(x_0,y_0)$ and $${z_0} =  - \frac{{{q_0}\left( {{x_0},{y_0}} \right) - {p_0}\left( {{x_0},{y_0}} \right)}}{{{q_1}\left( {{x_0},{y_0}} \right) - {p_1}\left( {{x_0},{y_0}} \right)}}$$ then $f(x_0,y_0,z_0)=0$, $g(x_0,y_0,z_0)=0$ and the quadrics defined by $f$ and $g$ have the same tangent plane at $(x_0,y_0,z_0)$.
\end{cor}
\begin{proof}
The tangent planes to $f=0$ and $g=0$ at $(x_0,y_0,z_0)$ are given by:
\[\frac{{\partial f}}{{\partial x}}\left( {{x_0},{y_0},{z_0}} \right)\left( {x - {x_0}} \right) + \frac{{\partial f}}{{\partial y}}\left( {{x_0},{y_0},{z_0}} \right)\left( {y - {y_0}} \right) + \frac{{\partial f}}{{\partial z}}\left( {{x_0},{y_0},{z_0}} \right)\left( {z - {z_0}} \right) = 0\]
\[\frac{{\partial g}}{{\partial x}}\left( {{x_0},{y_0},{z_0}} \right)\left( {x - {x_0}} \right) + \frac{{\partial g}}{{\partial y}}\left( {{x_0},{y_0},{z_0}} \right)\left( {y - {y_0}} \right) + \frac{{\partial g}}{{\partial z}}\left( {{x_0},{y_0},{z_0}} \right)\left( {z - {z_0}} \right) = 0\]
Previous theorem implies that the quadrics $f=0$ and $g=0$ have the same tangent plane at $(x_0,y_0,z_0)$. 
\end{proof}

Next we show how to compute easily the singular points of the cutcurve when they are in the line $p_1(x,y)=q_1(x,y)$. Next lemma will connect ${\bf S}_0(x,y)$ with the line $p_1(x,y)-q_1(x,y)=0$ and with $\Delta _{{\cal E}_1}(x,y)$ and $\Delta _{{\cal E}_2}(x,y)$ (the discriminants of ${\cal E}_1$ and ${\cal E}_2$, respectively).

\begin{lem}\label{cutcurveformula}
$${\bf S}_0(x,y)=\frac{1}{16}
\Big[
\left(p_1-q_1\right)^4+
\left(\Delta _{{\cal E}_1}-\Delta _{{\cal E}_2}\right)^2-
2\left(p_1-q_1\right)^2\left(\Delta _{{\cal E}_1}+\Delta _{{\cal E}_2}\right)
\Big].$$
\end{lem}

\begin{proof}
Since $\Delta _{{\cal E}_1}=p_1^2-4p_0$ and $\Delta _{{\cal E}_2}=q_1^2-4q_0$, we have
$$\frac{1}{16}
\Big[
\left(p_1-q_1\right)^4+
\left(\Delta _{{\cal E}_1}-\Delta _{{\cal E}_2}\right)^2-
2\left(p_1-q_1\right)^2\left(\Delta _{{\cal E}_1}+\Delta _{{\cal E}_2}\right))
\Big]=$$
$$={p_0}^2 - 2{p_0}{q_0} + {q_0}^2+{p_1}^2{q_0} + {q_1}^2{p_0} - {p_1}{q_0}{q_1} - {p_0}{p_1}{q_1}=
{\left( {{p_0} - {q_0}} \right)^2} - \left( {{p_1} - {q_1}} \right)\left( {{p_0}{q_1} - {q_0}{p_1}} \right)={\bf S}_0(x,y)$$ as desired.
\end{proof}

As a consequence of the equality in the previous lemma, we conclude that we can always compute the intersection points between the cutcurve and the line $p_1(x,y)-q_1(x,y)=0$ by just solving a degree two equation.

\begin{cor}\label{singpoints}
The intersection points between the cutcurve and the line $p_1(x,y)-q_1(x,y)=0$ are the same than the intersection points between the curve
$$\Delta _{{\cal E}_1}(x,y)-\Delta _{{\cal E}_2}(x,y)=0$$
and the line $p_1(x,y)-q_1(x,y)=0$. 
\end{cor}

In some cases the line $p_1(x,y)-q_1(x,y)=0$ (or part of it) is contained completely in the cutcurve: this can be checked directly by performing the corresponding substitution in ${\bf S}_0(x,y)$ or by using the previous corollary. In the particular case when $p_1(x,y)\equiv q_1(x,y)$ we have that $${\bf S}_0(x,y)=(p_0(x,y)- q_0(x,y))^2$$ and, as expected, all points in the cutcurve are singular. But this implies that the cutcurve is a conic ($p_0(x,y)-q_0(x,y)=0$) and all further computations (including the lifting) are greatly simplified.

Next proposition shows that the ``vertices" of the region ${\cal A}_{{\cal E}_1,{\cal E}_2}$ belonging to ${\bf S}_0(x,y)=0$ are always singular points of this curve.

\begin{propst}\label{propsingular}
If $(\alpha,\beta)$ is a point such that $${\bf S}_0(\alpha,\beta) = 0, \Delta _{{\cal E}_1}(\alpha,\beta)=0, \Delta _{{\cal E}_2}(\alpha,\beta)=0$$ then $p_1(\alpha,\beta)=q_1(\alpha,\beta)$ and $(\alpha,\beta)$ is a singular point of the cutcurve.
\end{propst}

\begin{proof}
From Lemma \ref{cutcurveformula} we have
$${\bf S}_0(\alpha,\beta) = 0, \Delta _{{\cal E}_1}(\alpha,\beta)=0, \Delta _{{\cal E}_2}(\alpha,\beta)=0\Longleftrightarrow \left(\frac{(p_1-q_1)^2}{4}\right)^2=0\Longleftrightarrow p_1=q_1$$ as desired.
\end{proof}

Finally, we show how to determine those singular points of the cutcurve not belonging to the line $p_1(x,y)=q_1(x,y)$. These points, according to Corollary \ref{cortangent}, come from tangential intersection points of the two considered quadrics and they are very easily lifted but more difficult to determine than those in the line $p_1(x,y)=q_1(x,y)$. To determine these points we have to solve the system of equations
$${\bf S}_0(x,y)=0,\quad\frac{\partial {\bf S}_0}{\partial x}(x,y)=0\quad\frac{\partial {\bf S}_0}{\partial y}(x,y)=0, \quad p_1(x,y)\ne q_1(x,y)$$
inside ${\cal A}_{{\cal E}_1,{\cal E}_2}$.
In order to solve this system of equations we need the polynomial containing the $x$-coordinates of the singular points of ${\bf S}_0(x,y)$ in the line $p_1(x,y)-q_1(x,y)=0$.

\begin{defin}\label{tau}
When $\deg(p_1(x,y)-q_1(x,y),y)=1$, the squarefree part of the polynomial obtained after replacing $y$ in 
$\Delta _{{\cal E}_1}(x,y)-\Delta _{{\cal E}_2}(x,y)$ by the value obtained by solving 
$p_1(x,y)-q_1(x,y)=0$ in terms of $y$ will be denoted by $\tau_{{\cal E}_1,{\cal E}_2}(x)$. Otherwise $\tau_{{\cal E}_1,{\cal E}_2}(x)=p_1(x,y)-q_1(x,y)$.
\end{defin}

Next theorem shows how to use subresultants to determine the tangential intersection points of the two considered quadrics whose projection is outside the line $p_1(x_0,y_0)-q_1(x_0,y_0)=0$.

\begin{thm}\label{tangencutcurve}
Let $(x_0,y_0)$ be a singular point of the cutcurve such that $p_1(x_0,y_0)\ne q_1(x_0,y_0)$. If
\begin{itemize}
\item $U_0(x)$ is the squarefree part of ${\bf Sres}_0\left({\bf S}_0,\frac{\partial {\bf S}_0}{\partial x};y\right)$.
\item $U_1(x,y)={\bf Sres}_1\left({\bf S}_0,\frac{\partial {\bf S}_0}{\partial x};y\right)=U_{11}(x)y+U_{10}(x)$.
\item $V_0(x)$ is the squarefree part of ${\bf Sres}_0\left({\bf S}_0,\frac{\partial {\bf S}_0}{\partial y};y\right)$.
\item $V_1(x,y)={\bf Sres}_1\left({\bf S}_0,\frac{\partial {\bf S}_0}{\partial y};y\right)=V_{11}(x)y+V_{10}(x)$.
\item $W(x)$ is the squarefree part of $U_{10}(x)V_{11}(x)-V_{10}(x)U_{11}(x)$.
\end{itemize}
then $x_0$ is a real root of the polynomial
$$\Omega_{{\cal E}_1,{\cal E}_2}(x)=\frac{\gcd(W(x),U_0(x),V_0(x))}{\gcd(W(x),\tau_{{\cal E}_1,{\cal E}_2}(x))}$$
and $$y_0=-U_{10}(x_0)/U_{11}(x_0)\quad \hbox{and/or}\quad y_0=-V_{10}(x_0)/V_{11}(x_0).$$
\end{thm}

\begin{proof} 
Since $(x_0,y_0)$ is a singular point of the cutcurve, we have the following equalities:
$$U_0(x_0)=0, U_1(x_0,y_0)=0, V_0(x_0)=0, V_1(x_0,y_0)=0.$$
By using $U_1(x_0,y_0)=0$ and $V_1(x_0,y_0)=0$  we have 
$$U_{10}(x_0)V_{11}(x_0)-V_{10}(x_0)U_{11}(x_0)=0$$ and that $x_0$ is a real root of 
$\gcd(U_{10}(x)V_{11}(x)-V_{10}(x)U_{11}(x),U_0(x),V_0(x))$
and a real root of $\gcd(W(x),U_0(x),V_0(x))$. Since $p_1(x_0,y_0)\ne q_1(x_0,y_0)$, according to Corollary \ref{singpoints} and Definition \ref{tau}, we have $\tau_{{\cal E}_1,{\cal E}_2}(x_0)\ne 0$ and this allows to conclude that $\Omega_{{\cal E}_1,{\cal E}_2}(x_0)=0$ as desired.
Note that, under these conditions, $U_{11}(x_0)\neq 0$ or $V_{11}(x_0)\neq 0$. If this is not the case then $(y-y_0)^2$ is a common factor of ${\bf S}_0(x_0,y)$, $\frac{\partial {\bf S}_0}{\partial x}(x_0,y)$ and $\frac{\partial {\bf S}_0}{\partial y}(x_0,y)$ but this is not possible since the degree in $y$ of ${\bf S}_0(x_0,y)$ is bounded by $4$.
\end{proof}

\subsection{Determining the intersection points of the cutcurve with the silhouette curves}
Next propositions provide the way to determine in a simpler way the points in the cutcurve which belong to each silhouette curve.

\begin{propst}\label{propregular1}
If ${\bf S}_0(\alpha,\beta)=0$ then
\[{\Delta _{{{\cal E}_1}}}(\alpha,\beta) = 0 \mathop  \Leftrightarrow \limits^{(I)} {\Delta _{{{\cal E}_2}}}(\alpha,\beta)= {\left( {{p_1}(\alpha,\beta) - {q_1}(\alpha,\beta)} \right)^2}\mathop  \Leftrightarrow \limits^{(II)}  2\left( {{q_0}(\alpha,\beta)+ {p_0}(\alpha,\beta)} \right) = {p_1}(\alpha,\beta){q_1}(\alpha,\beta)\]
\[{\Delta _{{{\cal E}_2}}}(\alpha,\beta) = 0 \mathop  \Leftrightarrow \limits^{(I)} {\Delta _{{{\cal E}_1}}}(\alpha,\beta) = {\left( {{p_1}(\alpha,\beta) - {q_1}(\alpha,\beta)} \right)^2}  \mathop  \Leftrightarrow \limits^{(II)} 2\left( {{q_0}(\alpha,\beta) + {p_0}(\alpha,\beta)} \right) = {p_1}(\alpha,\beta){q_1}(\alpha,\beta)\]
\end{propst}
\begin{proof}
Using Lemma \ref{cutcurveformula}, we have \[\Delta _{{{\cal E}_2}}^2 - 2{\left( {{p_1} - {q_1}} \right)^2}{\Delta _{{{\cal E}_2}}} + {\left( {{p_1} - {q_1}} \right)^4} = 0\quad \Longleftrightarrow\quad {\Delta _{{{\cal E}_2}}} = {\left( {{p_1} - {q_1}} \right)^2} \]
and we get the first equivalence. Using this one and ${\Delta _{{{\cal E}_2}}} = {q_1}^2 - 4{q_0}$, we conclude $p_1q_1=2(p_0+q_0)$. The second one is similar.
\end{proof}

As a consequence we have that solving each system 
$${\bf S}_0(x,y)=0,\quad  {\Delta _{{{\cal E}_i}}}(x,y)=0$$ 
is the same than solving the simpler system  
$$2(p_0(x,y)+q_0(x,y))=p_1(x,y)q_1(x,y), \quad  {\Delta _{{{\cal E}_i}}}(x,y)=0 .$$

An easy consequence of the previous proposition is the fact that the cutcurve and the silhouette curves have no common points outside the line $p_1(x,y) = q_1(x,y)$.
    
\begin{cor}
The system $${\bf S}_0(x,y)=0,\; p_1(x,y) \ne q_1(x,y),\; \Delta _{{\cal E}_1}(x,y)=0,\;  \Delta _{{\cal E}_2}(x,y)=0$$ has no real solutions.
\end{cor}

%\begin{proof} Since
%$$\Delta _{{\cal E}_1}=0 \Rightarrow p_1^2-4p_0=0 \Rightarrow p_0=\frac{p_1^2}{4},$$ $$\Delta _{{\cal E}_2}=0 \Rightarrow q_1^2-4q_0=0 \Rightarrow q_0=\frac{q_1^2}{4},$$
%by using Proposition \ref{propregular1}, we have
%$$2(q_0+p_0)=p_1q_1 \Rightarrow 2\left( {\frac{{{p_1}^2}}{4} + \frac{{{q_1}^2}}{4}} \right) = {p_1}{q_1} \Rightarrow {\left( {{p_1} - {q_1}} \right)^2} = 0.$$
%However, $(p_1-q_1)^2=0$ has no solutions because $p_1 \ne q_1$.
%\end{proof}

\subsection{Lifting the points of the cutcurve of ${\cal E}_1$ and ${\cal E}_2$}
Next it is shown how to perform the lifting of the points of the cutcurve.
\begin{thm}
If $\left(\alpha,\beta\right)$ is a point in the cutcurve such that $q_1(\alpha,\beta)\ne p_1(\alpha,\beta)$ then the z-coordinate of the point in the intersection curve is given by:
$$z=\frac{p_0(\alpha,\beta)-q_0(\alpha,\beta)}{q_1(\alpha,\beta)-p_1(\alpha,\beta)}.$$ 
If $(\alpha,\beta)$ is a point of the cutcurve such that $q_1(\alpha,\beta)=p_1(\alpha,\beta)$ then the lifting of this singular point can be made by using $g(\alpha,\beta,z)=0$ or $f(\alpha,\beta,z)=0$.
\end{thm}

\begin{proof} 
As we have seen
$${\bf S}_1(x,y;z)={\bf Sres}_1(f,g;z)=(q_1(x,y)-p_1(x,y))z+(q_0(x,y)-p_0(x,y)).$$
If $(\alpha,\beta)\in \Pi \left( {{\cal E}_1 \cap {\cal E}_1} \right)$ then
$${\bf S}_1(\alpha,\beta;z)=0\Leftrightarrow z=\frac{p_0(\alpha,\beta)-q_0(\alpha,\beta)}{q_1(\alpha,\beta)-p_1(\alpha,\beta)}$$ as desired.
\end{proof}

\section{Experimentation}
This section will present the experimentation performed together with some examples showing how to compute the intersection curve of two quadrics by using the results presented in the previous sections. In all examples two polynomials $f$ and $g$ define two quadrics ${\cal E}_1$ and ${\cal E}_2$  respectively. In these examples $${\cal A}_{{\cal E}_1,{\cal E}_2}=
\{(x,y)\in\R^2\colon \Delta _{{\cal E}_1}(x,y)\ge0,\; \Delta _{{\cal E}_2}(x,y)\ge0\}$$ defines the region where the cutcurve, defined by ${\bf S}_0(x,y)=0$, lives. Typically the lifting of the cutcurve will be made by using ${\bf S}_1(x,y;z)$. When one of the singular  points of the cutcurve can not be lifted by using ${\bf S}_1(x,y;z)$, we will  use $g(x,y,z)$ or $f(x,y,z)$ for that purpose (for those singular points outside the line $p_1(x,y)=q_1(x,y)$ we can use ${\bf S}_1(x,y;z)$ for performing the lifting too).

The way of proceeding will be always the following one:
\begin{enumerate}
\item Compute the implicit equation for the cutcurve ${\bf S}_0(x,y)$.
\item Compute the description of the region ${\cal A}_{{\cal E}_1,{\cal E}_2}$ where the cutcurve lives. These two polynomial inequalities will be used to determine which points on the projection we are going to consider.
\item Compute the singular points of the cutcurve ${\bf S}_0(x,y)$ in the line $p_1(x,y)=q_1(x,y)$: just solving a degree two univariate equation by Corollary \ref{singpoints}. Their lifting is made by using $f$ or $g$. This step requires to solve three univariate degree two equations
\item Compute the singular points of the cutcurve ${\bf S}_0(x,y)$ outside the line $p_1(x,y)=q_1(x,y)$: points coming from the projection of a tangential intersection whose lifting is made by using ${\bf S}_1(x,y;z)$ (computations guided by Theorem \ref{tangencutcurve}).
\item Compute the regular points of the cutcurve which are in the silhouette curves by using the Proposition \ref{propregular1} and their lifting by using ${\bf S}_1(x,y;z)$. These points will lead the discretisation of the cutcurve or the computation of the intervals where the parameterisation determined can be evaluated (since it might involve radicals) since they contain the points where the cutcurve ``starts" and ``stops" (if any). This step requiere to intersect two couples of conics.
\item Compute the branches of the cutcurve (always inside ${\cal A}_{{\cal E}_1,{\cal E}_2}$) by closed formulae involving radicals or discretising them (and their lifting by using ${\bf S}_1(x,y;z)$). Either solving degree two equations of computing numerically the solutions of the, at most degree four, univariate equation ${\bf S}_0(\alpha,y)$ for several values of $\alpha$.
\end{enumerate}

In the next four examples, the lifting of the cutcurve will be made after discretising the regular branches of the cutcurve (always inside ${\cal A}_{{\cal E}_1,{\cal E}_2}$).  All examples include computations made with Maple. In the (right) next figures the cutcurve is drawn in black, silhouette curves are drawn in blue (and in the last example in red and blue) and the line $p_1(x,y)=q_1(x,y)$ will be always dotted. The admissible region is always the darker region.

\begin{examp}\label{third}{\rm 
Let $f$ and $g$ be the polynomials
$$f(x,y,z)=z^2+(-6x-y-1)z-9x^2-3xy+4y^2+9x-9y-2$$ $$g(x,y,z)=z^2-2z+x^2-3y^2+9x-2y+6$$
defining two ellipsoids ${\cal E}_1$ and ${\cal E}_2$, whose intersection curve is to be computed. In this case we have:
$${\bf S}_0(x,y)=136x^4+72x^3y-238x^2y^2-78xy^3+46y^4+432x^3+230x^2y-15xy^2-108y^3$$
$$+249x^2+204xy-28y^2+33x+100y+54$$
$${\cal A}_{{\cal E}_1,{\cal E}_2}=\left\{ (x,y) \in \mathbb{R}^2: 
 \begin{array}{l}
72x^2+24xy-15y^2-24x+38y+9\ge0,\\ -4x^2+12y^2-36x+8y-20\ge0
\end{array} \right\}.$$
From Corollary \ref{singpoints}, the singular points of the cutcurve in the line $p_1(x,y)=q_1(x,y)$ are determined by solving:
$$76x^2+24xy-27y^2+12x+30y+29=0 \wedge   6x + y + 1 =   2  .$$
In this way we get the singular points
$$A=\left( {\frac{9}{{104}} + \frac{{\sqrt {10345} }}{{520}},\frac{{25}}{{52}} - \frac{{3\sqrt {10355} }}{{260}}} \right)\qquad
B=\left( {\frac{9}{{104}} - \frac{{\sqrt {10345} }}{{520}},\frac{{25}}{{52}} + \frac{{3\sqrt {10355} }}{{260}}} \right) .$$ The first one, $A$, is an isolated point of the cutcurve but outside ${\cal A}_{{\cal E}_1,{\cal E}_2}$ and will not be considered. The second one, $B$, will be lifted by using $f(x,y,z)=0$ or $g(x,y,z)=0$ providing two different points in the intersection curve.

Point $C=(-0.5989698028,-0.6502822952)$ is common to ${\bf S}_0(x,y)=0$ and $\Delta_{{\cal E}_1}(x,y)=0$ and, by using Proposition \ref{propregular1}, was determined by solving:
$$-16x^2-6xy+2y^2+24x-24y+6=0 \wedge -4x^2+12y^2-36x+8y-20=0.$$

Points $D=(-2.336955328,-6.163216205)$ and $E=(21.765280490,-32.199082657)$ are common to  ${\bf S}_0(x,y)=0$ and $\Delta_{{\cal E}_1}(x,y)=0$ and, from Proposition \ref{propregular1}, were determined by solving:
$$-16x^2-6xy+2y^2+24x-24y+6=0 \wedge 72x^2+24xy-15y^2-24x+38y+9=0$$

The lifting of $\Pi({\cal E}_1\cap{\cal E}_2)$, outside the singular points of the cutcurve, will be made by using
${\bf S}_1(x,y;z)$: $$z = \frac{10x^2+3xy-7y^2+7y+8}{-6x-y+1}.$$

\begin{figure}
\centering
\fbox{\includegraphics[scale=0.51]{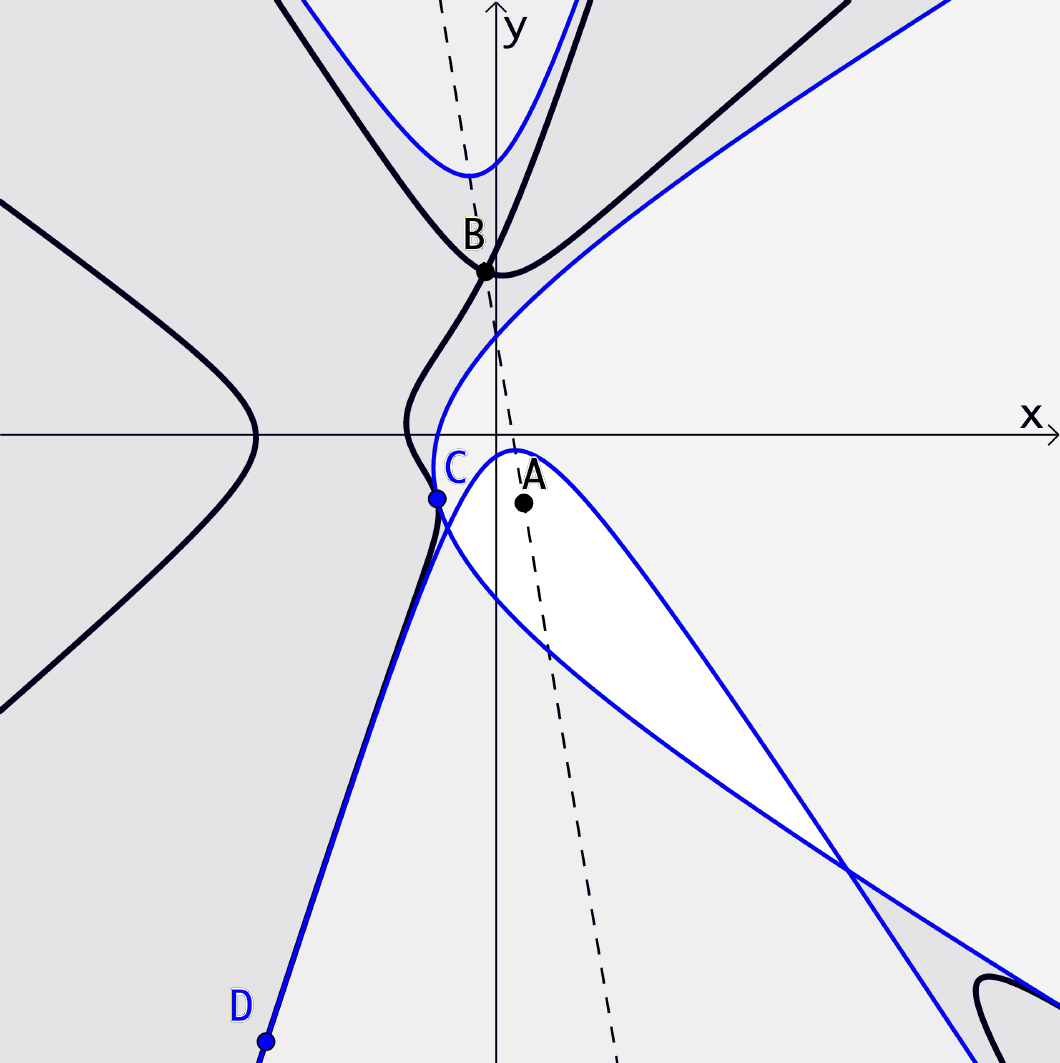}}\hskip 2cm
\fbox{\includegraphics[scale=0.55]{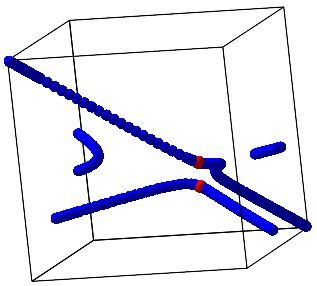}}
\end{figure}}
\end{examp}

\begin{examp}\label{fourth}{\rm 
Let $f$ and $g$ be the polynomials
$$f(x,y,z)=z^2+xz+y \qquad g(x,y,z)=z^2+yz+x$$
defining two hyperbolic paraboloids, ${\cal E}_1$ and ${\cal E}_2$, whose intersection curve is to be computed. In this case we have
$${\bf S}_0(x,y)=(x-y)^2(x+y+1)$$
and 
$$
{\cal A}_{{\cal E}_1,{\cal E}_2}=\{(x,y)\in\R^2\colon x^2-4y\ge 0,\; y^2-4x\ge 0 \}.$$

All the points in the line $p_1(x,y)=q_1(x,y)$ belong to the cutcurve and they are singular: this can be checked  by solving (according to Corollary \ref{singpoints}):
$$x^2-y^2+4x-4y=0 \wedge x-y=0.$$
In this case, there is one special singular point, $D=(-1/2,-1/2)$, that has been already classified since it belongs to the line $p_1(x,y)=q_1(x,y)$.

Points $O=(0,0)$, $B=(1,-2)$ and $C=(4,4)$ are common to ${\bf S}_0(x,y)=0$ and $\Delta_{{\cal E}_1}(x,y)=0$ and, from Proposition \ref{propregular1}, were determined by solving
$$-xy+2x+2y=0 \wedge y^2-4x=0.$$
On the other hand, points $A=(-2,1)$, $O=(0,0)$ and $C=(4,4)$ are common to ${\bf S}(x,y)=0$ and $\Delta_{{\cal E}_2}(x,y)=0$ and, from Proposition \ref{propregular1}, were determined by solving
$$-xy+2x+2y=0 \wedge x^2-4y=0.$$
Note that points $O$ and $C$ are common to ${\bf S}_0=0$, $\Delta_{{\cal E}_1}=0$, $\Delta_{{\cal E}_2}=0$ and, as seen in Propostion \ref{propsingular}, belong to the line $x-y=0$.

The lifting of $\Pi({\cal E}_1\cap{\cal E}_2)$, outside the singular points of the cutcurve, will be made by using
${\bf S}_1(x,y;z)$: $$z=\frac{x-y}{x-y} = 1.$$

Singular points of the cutcurve will be lifted by using $g(x,y,z)$:
$$g(x,y,z)=0\Leftrightarrow z =  - \frac{y}{2} + \frac{{\sqrt {{y^2} - 4x} }}{2} \vee z =  - \frac{y}{2} - \frac{{\sqrt {{y^2} - 4x} }}{2}.$$

To characterise the intersection curve of  ${\cal E}_1$ and ${\cal E}_2$ we must determine the cutcurve ${\bf S}_0(x,y)=0$ and its lifting. We will use the following functions: 
\begin{itemize}
    \item For $x \in \mathbb{R}$, we define:
    $$h_1(x)=x \qquad h_2(x)=-x-1$$
    \item Let $e_1$ and $e_2$ be the functions defined by:
    $${e_1}\left( {x,y} \right) =  - \frac{y} {2} + \frac{{\sqrt {{y^2} - 4x} }} {2} \qquad {e_2}\left( {x,y} \right) =  - \frac{y}{2} - \frac{{\sqrt {{y^2} - 4x} }}{2}$$
\end{itemize}
The parameterisation of the intersection curve is given by the following three components:
\begin{itemize}
    \item For $x \in \left] { - \infty ,0} \right] \cup \left[ {4, + \infty } \right[$: $(x,h_1(x),e_1(x,h_1(x))$.
 \item For $x \in \left] { - \infty ,0} \right] \cup \left[ {4, + \infty } \right[$:
    $(x,h_1(x),e_2(x,h_1(x))$.
    \item For $x \in \mathbb{R}$: $(x,h_2(x),1)$.
\end{itemize}
The intervals are determined by analysing the radical expressions and solving the corresponding (univariate) polynomial inequalities.

By using the QI online computation server (available at \url{https://gamble.loria.fr/qi/server/}), we get a parameterisation of the intersection curve without involving radicals. Instead here we show how to get a parameterisation of the intersection curve by solving several degree two equations.

\begin{minipage}{.5\linewidth}
 \centering
 \fbox{\includegraphics[width=0.7\linewidth]{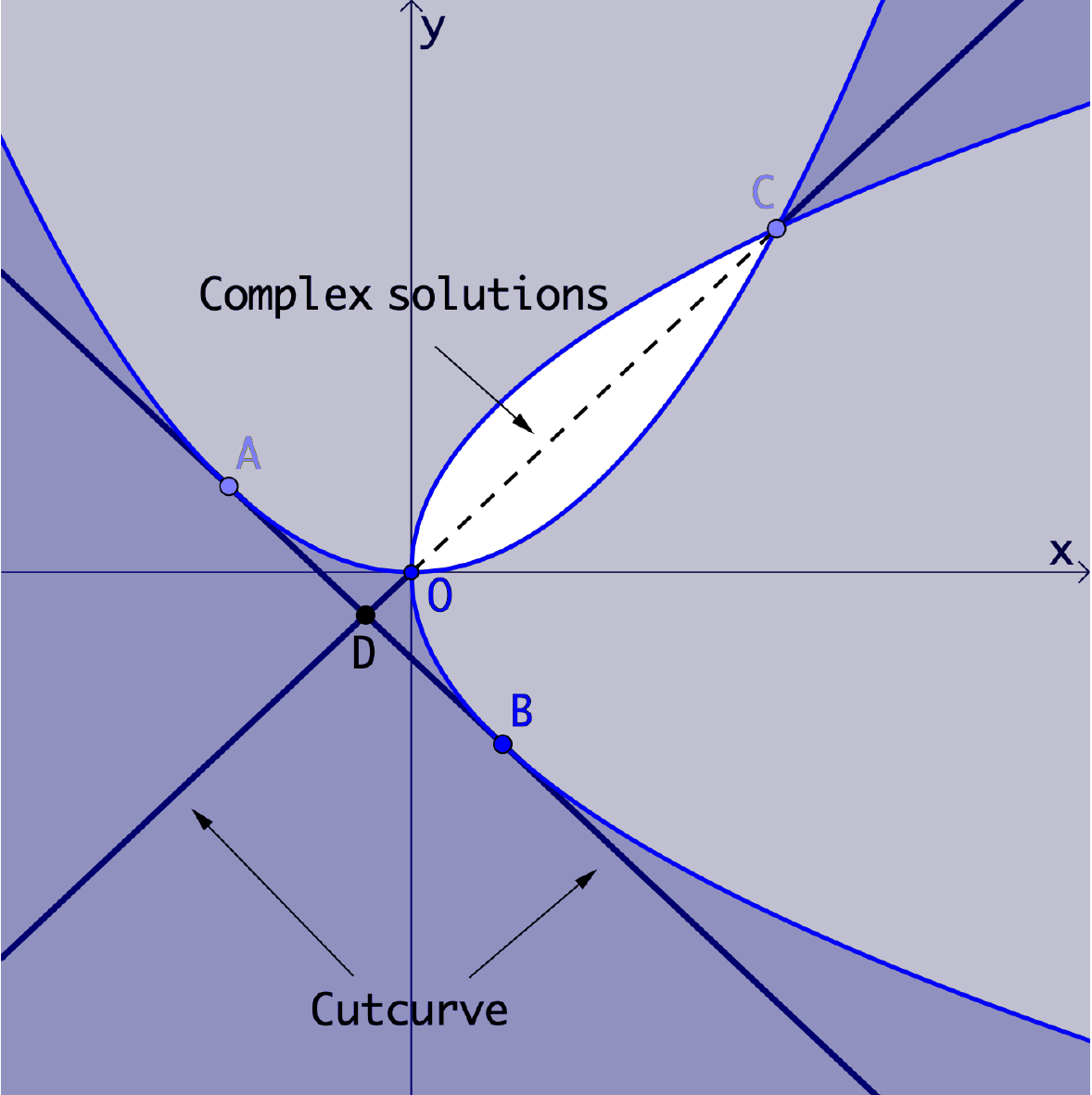}}
\end{minipage}
\begin{minipage}{.5\linewidth}
 \centering
\fbox{\includegraphics[width=0.75\linewidth]{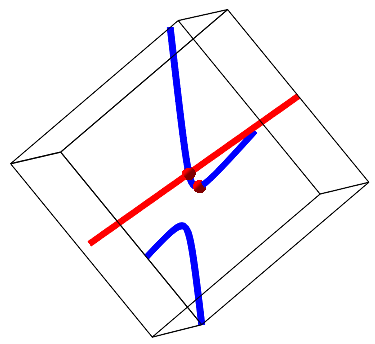}}
\end{minipage}}
\end{examp}

\begin{examp}\label{fifth}{\rm 
Let $f$ and $g$ be the polynomials
$$f(x,y,z)=z^2+(y-2x-1)z-x^2-y^2-xy+x-y+1$$ $$g(x,y,z)=z^2+(x-y)z+x^2+y^2-xy-2x+y-5$$
defining one hyperboloid of one sheet and one ellipsoid, ${\cal E}_1$ and ${\cal E}_2$, whose intersection curve is to be computed. In this case we have
$${\bf S}_0(x,y)=7x^4-11x^3y+23x^2y^2-6xy^3+4y^4-17x^3+12x^2y-8xy^2+6y^3-50x^2+26xy$$
$$-29y^2+10x-9y+31$$
$$
{\cal A}_{{\cal E}_1,{\cal E}_2}=\{(x,y)\in\R^2\colon 8x^2+5y^2+2y-3\ge 0,\; -3x^2+2xy-3y^2+8x-4y+20\ge 0 \}.$$

From Corollary \ref{singpoints}, singular points of the cutcurve in the line $p_1(x,y)=q_1(x,y)$ were determined by solving:
\[11x^2-2xy-8x+8y^2+6y-23 \wedge -3x+2y-1= 0\]
They are $$A=\left( { - \frac{3}{{13}} - \frac{{3\sqrt {14} }}{{13}},\frac{2}{{13}} - \frac{{9\sqrt {14} }}{{26}}} \right)\qquad B=\left( { - \frac{3}{{13}} + \frac{{3\sqrt {14} }}{{13}},\frac{2}{{13}} + \frac{{9\sqrt {14} }}{{26}}} \right)$$ and can be lifted by using $g(x,y,z)=0$ or $f(x,y,z)=0$ producing four (singular) points in the intersection curve.

Common points to ${\bf S}_0(x,y)=0$ and $\Delta_{{\cal E}_2}(x,y)=0$ can be determined, according to Proposition \ref{propregular1}, by solving
$$2x^2-7xy+y^2-x-y-8=0 \wedge -3x^2+2xy-3y^2+8x-4y+20=0 .$$
These points are $C=(-1.468654,0.233082)$, $D=(-0.575622,1.494633)$, $E=(-0.285566,-3.292475)$ and $F=(4.341889,0.829820)$. The lifting of $\Pi({\cal E}_1\cap{\cal E}_2)$, outside the singular curves of the cutcurve, will be made by using
${\bf S}_1(x,y;z)$: $$z = -\frac{2x^2+2y^2-3x+2y-6}{3x-2y+1}.$$

\begin{minipage}{.5\linewidth}
 \centering
\fbox{\includegraphics[width=0.7\linewidth]{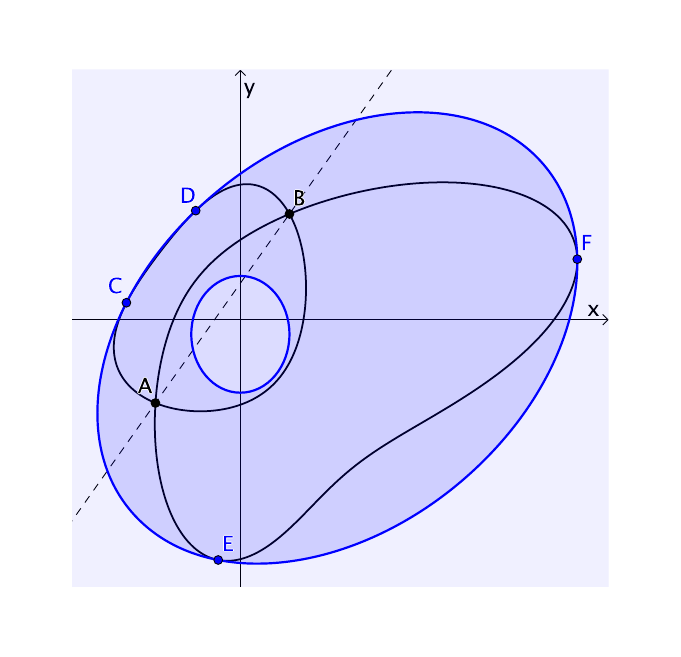}}
\end{minipage}
\begin{minipage}{.5\linewidth}
 \centering
\fbox{\includegraphics[width=0.7\linewidth]{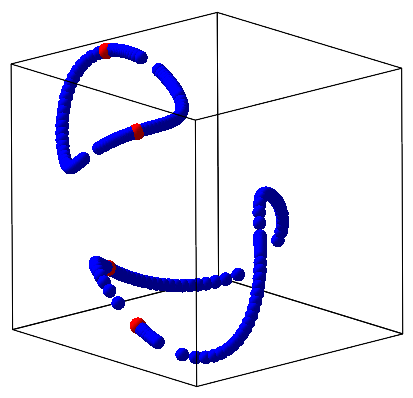}}
\end{minipage}}
\end{examp}

Last example shows a situation where the introduced technics are specially useful to determine the intersection curve of the two considered quadrics.

\begin{examp}\label{sixth}{\rm 
 Let $f$ and $g$ be the polynomials
$$f(x,y,z)={z^2} + \left( { - \frac{2}{3}x + \frac{2}{3}y} \right)z + \frac{{{x^2}}}{3} + \frac{{{y^2}}}{3} - \frac{1}{3}$$ $$g(x,y,z)={z^2} + \left( {\frac{{ - 2x + 24y - 2}}{{17}}} \right)z + \frac{{{x^2}}}{{17}} + \frac{{2x}}{{17}} - \frac{3}{{17}} + \frac{{12{y^2}}}{{17}}$$
defining two ellipsoids, ${\cal E}_1$ and ${\cal E}_2$, whose intersection curve is to be computed. In this case we have:
$${\bf S}_0(x,y)=\frac{{196}}{{2601}}{x^4} + \frac{{616}}{{2601}}{x^3}y + \frac{{920}}{{2601}}{x^2}{y^2} + \frac{{836}}{{2601}}x{y^3} + \frac{{361}}{{2601}}{y^4} - \frac{{112}}{{2601}}{x^3} - \frac{{56}}{{867}}{x^2}y - \frac{{112}}{{2601}}x{y^2}$$
$$- \frac{{76}}{{2601}}{y^3} - \frac{{104}}{{867}}{x^2} - \frac{{632}}{{2601}}xy - \frac{{368}}{{2601}}{y^2} + \frac{{176}}{{2601}}x + \frac{{184}}{{2601}}y + \frac{{52}}{{2601}}$$
$$
{\cal A}_{{\cal E}_1,{\cal E}_2}=\left\{(x,y)\in\R^2\colon
-2x^2-2xy-2y^2+3 \ge 0,
-4x^2-6xy-15y^2-8x-6y+13 \ge 0
\right\}$$

Singular points of the cutcurve in the line $p_1(x,y)=q_1(x,y)$ can be determined, as seen in Corollary \ref{singpoints}, by solving
$$-434x^2-362xy-38y^2+288x+216y+399=0\wedge -14x-19y+3=0$$
They are $$A=\left( {\frac{3}{{14}} - \frac{{11\sqrt {95} }}{{70}},\frac{{11\sqrt {95} }}{{95}}} \right)\qquad C=\left( {\frac{3}{{14}} + \frac{{11\sqrt {95} }}{{70}},-\frac{{11\sqrt {95} }}{{95}}} \right)$$ and are outside ${\cal A}_{{\cal E}_1,{\cal E}_2}$: therefore they will not be lifted. Moreover, the cutcurve has a third singular (and isolated) point $B=(1,0)$ which is inside ${\cal A}_{{\cal E}_1,{\cal E}_2}$ but not in the line $p_1(x,y)=q_1(x,y)$. This means that $B$ is the projection of a tangential intersection point of ${\cal E}_1$ and ${\cal E}_2$ and it has been computed by using Theorem \ref{tangencutcurve}.

Points $$D=(-1.310086292,1.116297338)\qquad E=(-1.032926046,-0.320076179)$$ are common to ${\bf S}_0(x,y)=0$ and $\Delta_{{\cal E}_1}(x,y)=0$ and, from Proposition \ref{propregular1}, were determined by solving
$$18x^2+26xy+29y^2+4x+2y-26 = 0 \wedge -2x^2-2xy-2y^2+3 = 0  .$$

On the other hand, the points $$F=(-1.310059433,1.116308957) \qquad G=(0.4229961827, -1.105788551)$$ are common to ${\bf S}_0(x,y)=0$ and $\Delta_{{\cal E}_2}(x,y)=0$ and, from Proposition \ref{propregular1}, were determined by solving
$$18x^2+26xy+29y^2+4x+2y-26 = 0 \wedge -4x^2-6xy-15y^2-8x-6y+13= 0  .$$

Figure \ref{AFD_2} (left) shows the location of all these points with respect to the cutcurve and the silhouette curves.  In order to determine what is the relative position of the points $A$, $F$ and $D$ with respect to this three curves we need to use the results introduced in the previous section. Figure \ref{AFD_2} (center and right) shows in detail what is happening in that area.

\begin{figure}[hbt]
\centering
\fbox{\includegraphics[scale=0.15]{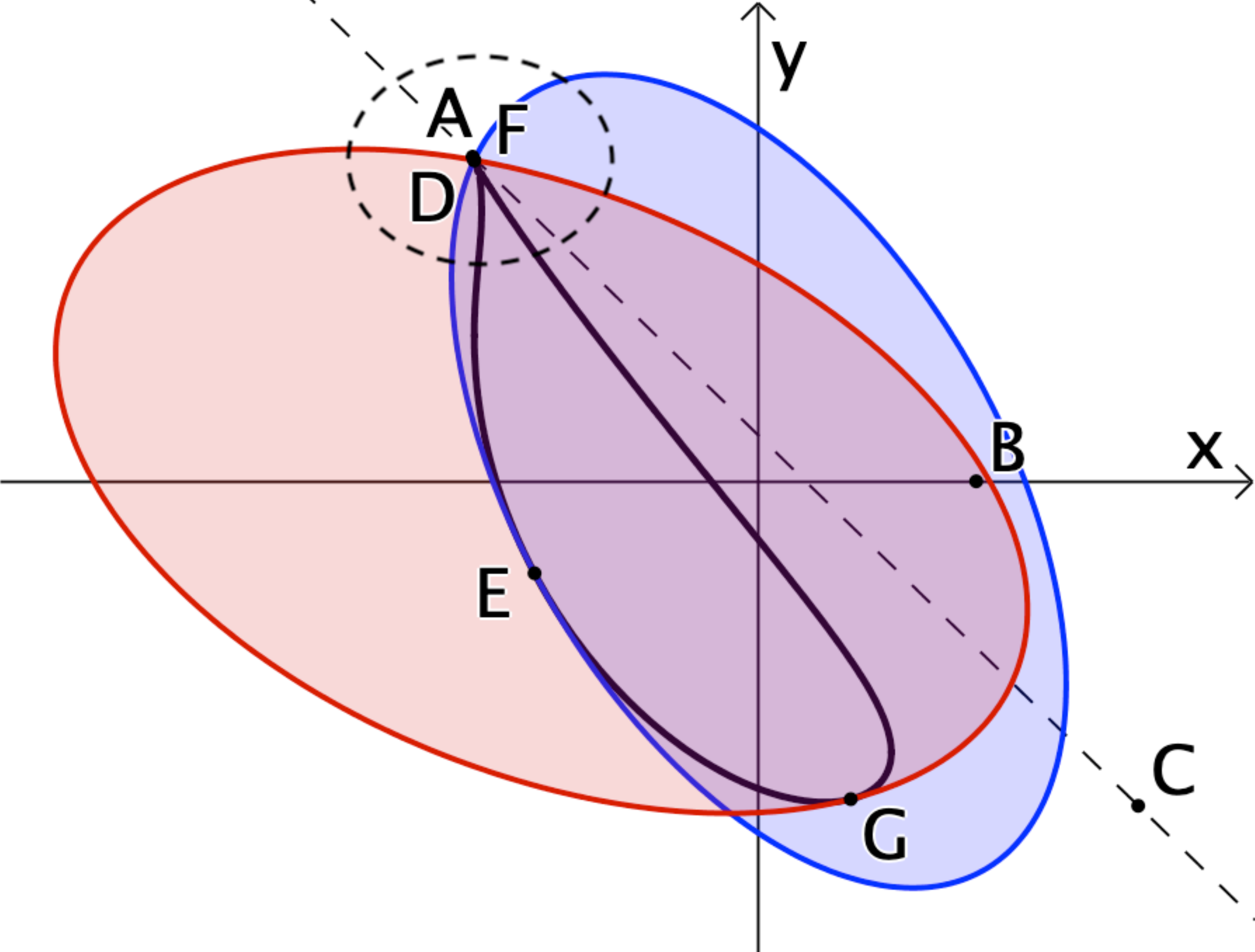}}\hskip .75cm 
\fbox{\includegraphics[scale=.8]{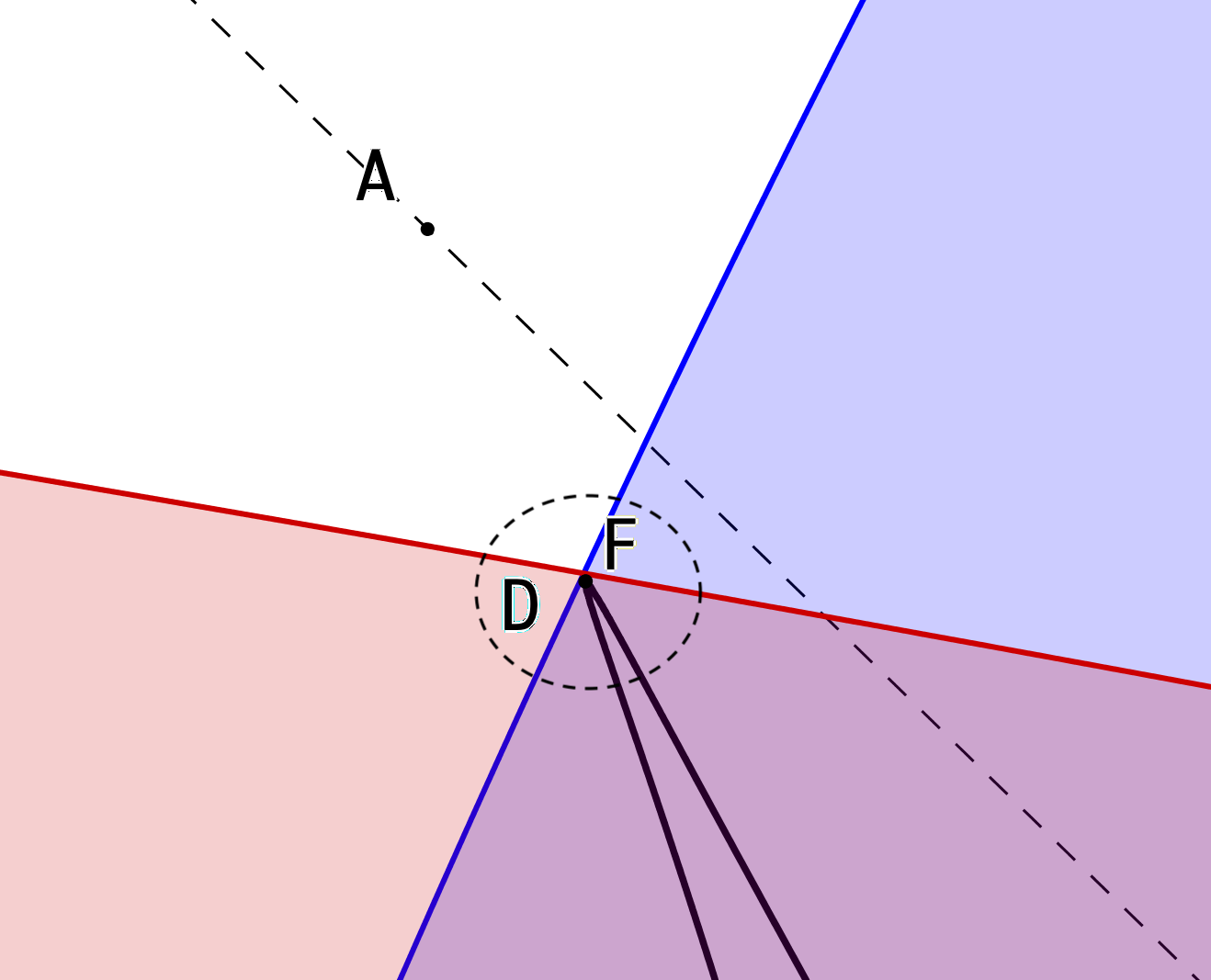}}\hskip .75cm 
\fbox{\includegraphics[scale=0.15]{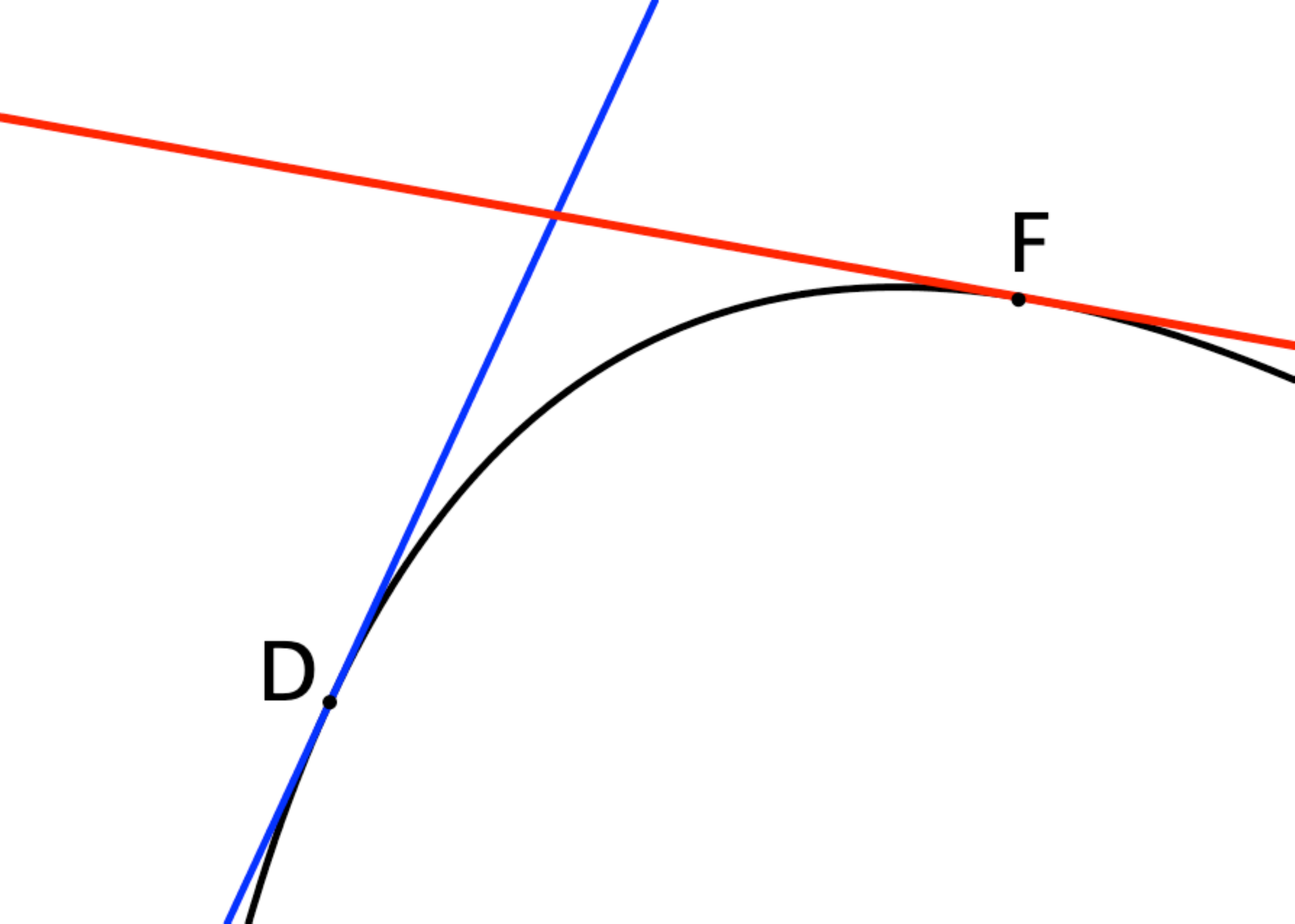}}
\caption{\label{AFD_2}Cutcurve and silhouette curves showing where points $A$, $F$ and $D$ are (left). Location of points $A$, $F$ and $D$  with respect to the cutcurve and the silhouette curves (center and right).}
\end{figure}

Point $A$ is outside the region ${\cal A}_{{\cal E}_1,{\cal E}_2}$ and does not play any role when computing the intersection curve between the two considered quadrics. Points $F$ and $D$ does belong to the intersection between the cutcurve and one of the silhouette curves. Proposition \ref{propsingular} helps to conclude that the intersection between the cutcurve and the two silhouette curves is empty and not $A$ as Figure \ref{AFD_2} (left) could suggest.

The lifting of the regular points of the cutcurve and of the point $B$ will be made by using
${\bf S}_1(x,y;z)$: $$z = \frac{14x^2-19y^2-6x-8}{28x+38y-6}.$$
Note that $B$ is the projection of the point $(1,0,0)$ where the ellipsoids are tangent and this is the reason why it can be lifted by using ${\bf S}_1(x,y;z)$ (as seen in Theorem \ref{singular2}). The intersection curve of the two considered quadrics can be found in Figure \ref{AFD_3}.

\begin{figure}[H]
\centering
\fbox{\includegraphics[scale=0.5]{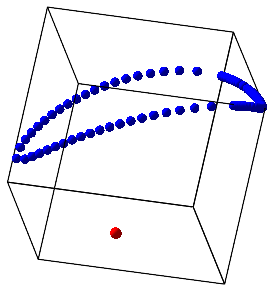}}
\caption{\label{AFD_3}Two ellipsoids with a curve and an isolated point (tangent to both ellipsoids) in common.}
 \end{figure}}
\end{examp}

The algorithm presented here has been fully implemented in Maple. Table \ref{table} shows its behaviour when applied to a database of $50$ examples most of them taken from the QI online computation server available at \url{https://gamble.loria.fr/qi/server/}. In the annex the equations of each pair of quadrics used can be found.

Figure \ref{maple} shows the shape of the output of this implementation when applied to a particular case. In this concrete case the intersection curve is not discretised and it is presented by a parameterisation involving radicals but the output shows the intervals where $x$ can be evaluated. For this concrete example the cutcurve has two singular points inside the admissible region and there are no tangential intersections whose projection is outside the line $p_1(x,y)=q_1(x,y)$.

\begin{figure}[hbt]
\centering
\fbox{\includegraphics[scale=0.55]{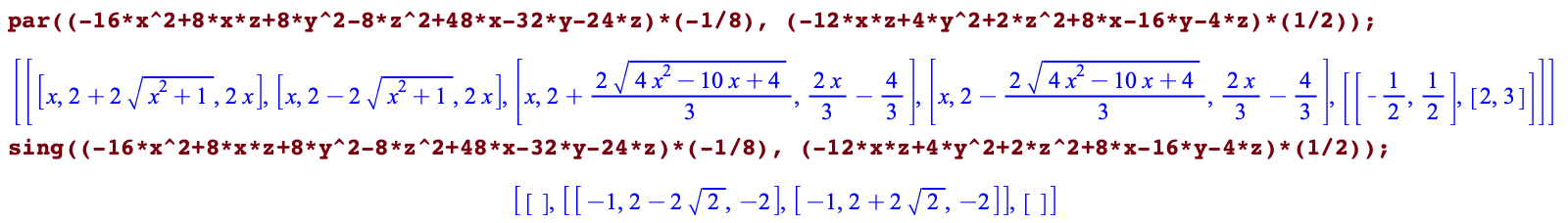}}
\caption{\label{maple}Using Maple to compute the intersection curve of two quadrics.}
 \end{figure}

The meaning of the columns in Table \ref{table} is the following: 
\begin{itemize} 
\item Third column reflects if the intersection curve has been discretised or not. 
\item Forth column shows the number of points computed when the intersection curve has been discretised. 
\item Fifth column shows if there are tangential intersection points whose projection is outside the line $p_1(x,y)=q_1(x,y)$
\end{itemize}

\begin{table}[H]\scriptsize
\centering
\begin{tabular}{|c|c|c|c|c|}
\hline
\textbf{Example} & \textbf{Time}  & \textbf{Discretisation} & \textbf{Number of Points} & \textbf{Tangential intersections}\\ \hline
Example 1        & 0.395          & yes      &     40     &       no     \\ \hline
Example 2        & 1.295          & yes      &      242    &      no         \\ \hline
Example 3        & 1.380          & yes      &     244     &     no          \\ \hline
Example 4        & 0.313          & no     &          &        no       \\ \hline
Example 5        & 1.586          & yes     &    476      &     no           \\ \hline
Example 6        & 1.502          & yes     &    472    &        no        \\ \hline
Example 7        & 1.338          & yes     &     362     &        no        \\ \hline
Example 8        & 1.290          & yes     &      202    &      yes          \\ \hline
Example 9       & 0.822          & yes     &   80   &         no       \\ \hline
Example 10       & 3.791          & yes     &   840       &     no           \\ \hline
Example 11       & 0.193          & no     &          &          no      \\ \hline
Example 12       & 3.387          & yes    &      762    &        no         \\ \hline
Example 13       & 0.236          & no     &          &           no       \\ \hline
Example 14       & 0.353          & no     &          &           no     \\ \hline
Example 15       & 0.357          & no     &          &       no         \\ \hline
Example 16       & 0.292          & no      &          &      no          \\ \hline
Example 17       & 0.330          & no     &          &       no          \\ \hline
Example 18       & 0.294          & yes      &    40      &     no          \\ \hline
Example 19       & 0.224          & no      &          &        no        \\ \hline
Example 20       & 0.372          & no     &          &         no        \\ \hline
Example 21       & 3.114          & yes     &   1114       &     no           \\ \hline
Example 22       & 0.215          & no      &          &        no        \\ \hline
Example 23       & 0.247          & no      &          &        no       \\ \hline
Example 24       & 0.616          & no      &          &         no       \\ \hline
Example 25       & 0.265          & no      &          &        no        \\ \hline
Example 26       & 0.334          & no     &          &         no       \\ \hline
Example 27       & 3.541          & yes    &     558     &      no          \\ \hline
Example 28       & 0.296          & no     &          &          no       \\ \hline
Example 29       & 1.664          & yes    &    472      &       no          \\ \hline
Example 30       & 0.181          & no   &          &            no      \\ \hline
Example 31       & 0.143          & no    &          &           no       \\ \hline
Example 32       & 0.200          & no    &          &            no      \\ \hline
Example 33       & 0.290          & no   &          &           no        \\ \hline
Example 34       & 0.234          & no    &          &       no           \\ \hline
Example 35       & 0.273          & no    &          &       no           \\ \hline
Example 36       & 0.329          & no    &          &       no           \\ \hline
Example 37       & 0.351          & no   &          &        no           \\ \hline
Example 38       & 0.237          & no   &          &         no          \\ \hline
Example 39       & 0.242          & no   &          &         no          \\ \hline
Example 40       & 0.234          & no   &          &         yes          \\ \hline
Example 41       & 0.272          & no  &          &         yes         \\ \hline
Example 42       & 0.200          & no  &          &        yes           \\ \hline
Example 43       & 0.332          & no   &          &        no           \\ \hline
Example 44       & 0.899          & yes  &     160     &      no             \\ \hline
Example 45       & 0.310          & no  &          &        no            \\ \hline
Example 46       & 0.304          & no &         &         yes          \\ \hline
Example 47       & 1.165          & no   &          &        yes           \\ \hline
Example 48       & 5.835          & yes  &    1574  &        yes           \\ \hline
Example 49       & 6.542          & yes  &   1894       &     yes              \\ \hline
Example 50       & 1.195         &  yes &    120     &      yes             \\ \hline
\end{tabular}
\caption{\label{table} Results from applying the Maple implementation to the quadrics in the Annex.}
\end{table}

Its practical behaviour is quite good and admits several improvements specially the step where the intersection curve is discretised. At this moment this is made by solving as many times as needed the equation $S_0(\alpha,y)=0$ for $\alpha$ belonging to the different intervals provided by the singular points and critical points of the curve $S_0(x,y)=0$. It must be noted here that the equation $S_0(\alpha,y)=0$ is never solved when $\alpha$ is the projection of a singular point of $S_0(x,y)=0$ since these points are previously computed by using the results in subsection \ref{lifting1}.

\section{Conclusions}
We have introduced a new approach to deal with the computation of the intersection curve of two quadrics. The main ingredients of this approach are a detailed analysis of the cutcurve, its singular points and of its relation with the silhouette curves together with the using of an uniform way to perform the lifting of the cutcurve to the intersection curve of the two considered quadrics.

Concerning the analysis of the cutcurve we classify its singular points in two different types depending on how they will be lifted. Those belonging to the line $p_1(x,y)=q_1(x,y)$ are easy to compute and difficult to lift (but just solving a degree two equation) and those not in that line which are more complicated to be determined but easier to lift.

This approach is not intended to classify the intersection curve between the two considered quadrics. Its main goal is to produce in a very direct way a description of the intersection curve which is topologically correct. This is the reason why we allow in the lifting of the cutcurve, when possible, the use of radicals or we rely on the discretisation of the branches of the cutcurve (uniquely determined by the points computed in that curve).

The algorithm has been fully implemented in Maple showing a very good practical behaviour.

\bibliographystyle{elsarticle-num}
% \bibliographystyle{elsarticle-harv}
% \bibliographystyle{elsarticle-num-names}
% \bibliographystyle{model1a-num-names}
% \bibliographystyle{model1b-num-names}
% \bibliographystyle{model1c-num-names}
% \bibliographystyle{model2-names}
% \bibliographystyle{model3a-num-names}
% \bibliographystyle{model3-num-names}
% \bibliographystyle{model4-names}
% \bibliographystyle{model5-names}
% \bibliographystyle{model6-num-names}

%\bibliography{sample}

\begin{thebibliography}{88}
\bibitem{Basu2006}
Basu, S., Pollack, R., Roy, M.-F (2006).
Algorithms in Real Algebraic Geometry. Springer--Verlag, http://doi.org/10.1007/3-540-33099-2.

\bibitem{Berberich:2005:ECE:1064092.1064110}
Berberich, E., Hemmer, M., Kettner, L., Schömer, E., \& Wolpert, N. (2005). An exact, complete and efficient implementation for computing planar maps of quadric intersection curves. In Proceedings of the Twenty-first32
Annual Symposium on Computational Geometry, SCG ’05, ACM, New York, NY, USA, pp. 99–-106. http://doi.org/10.1145/1064092.1064110.


\bibitem{DUPONT2008168}
Dupont, L., Lazard, D., Lazard, S., \& Petitjean, S. (2008). Near-optimal parameterization of the intersection of quadrics: I. The generic algorithm. Journal of Symbolic Computation 43 (3), 168-191. http://doi.org/10.1016/j.jsc.2007.10.006

\bibitem{Dupont2008c}
Dupont, L., Lazard, D., Lazard, S., \& Petitjean, S. (2008). Near-optimal parameterization of the intersection of quadrics: II. A classification of pencils. Journal of Symbolic Computation 43 (3), 192--215. http://doi.org/10.1016/j.jsc.2007.10.012

\bibitem{Dupont2008c3}
Dupont, L., Lazard, D., Lazard, S., \& Petitjean, S. (2008). Near-optimal parameterization of the intersection of quadrics: III. Parameterizing singular intersections. Journal of Symbolic Computation 43 (3), 216--232. http://doi.org/10.1016/j.jsc.2007.10.007

\bibitem{Farouki1989}
Farouki, R. T., Neff, C., \& O’Conner, M. A. (1989). Automatic parsing of degenerate quadric-surface intersections. ACM Transactions on Graphics 8 (3), 174--203. http://doi.org/10.1145/77055.77058

\bibitem{Geismann:2001:CCA:378583.378689}
Geismann, N., Hemmer, M., Sch{\"o}mer, E. (2001). Computing a 3-dimensional cell in an arrangement of quadrics:  Exactly and actually!, in:  Proceedings of the Seventeenth Annual Symposium on ComputationalGeometry, SCG ’01, ACM, New York, NY, USA, pp. 264–-273.  doi:10.1145/378583.378689.

\bibitem{Goldman:1991:CAR:112515.112545}
Goldman, R. N., \& Miller, J. R. (1991). Combining algebraic rigor with geometric robustness for the detection and calculation of conic sections in the intersection of two natural quadric surfaces, in: Proceedings of the First ACM Symposium on Solid Modeling Foundations and CAD/CAM Application, SMA '91, ACM, New York, NY, USA, pp. 221--231. http://doi.acm.org/10.1145/112515.112545

\bibitem{Gonzalez-Vega2009}
 Gonzalez-Vega, L.,  R\'ua, I. (2009), Solving the implicitization, inversion and reparametrization problems for rational curves through subresultants,  Computer  Aided  Geometric  Design  26  (9) 941–-961. https://doi.org/10.1016/j.cagd.2009.07.003

\bibitem{JOHNSTONE1992179}
Johnstone, J. K., \& Shene, C. K. (1992). Computing the intersection of a plane and a natural quadric. Computers and Graphics 16 (2), 179--186. http://doi.org/10.1016/0097-8493(92)90045-W

\bibitem{LAZARD200674}
Lazard, S., Peñaranda, L. M.,\& Petitjean, S. (2006). Intersecting quadrics: An efficient and exact implementation. In Computational Geometry 35 (1), 74--99. http://doi.org/10.1016/j.comgeo.2005.10.004

\bibitem{LEVIN197973}
Levin, J. Z. (1979). Mathematical models for determining the intersections of quadric surfaces. Computer Graphics and Image Processing 11 (1), 73--87. http://doi.org/10.1016/0146-664X(79)90077-7

\bibitem{Levin:1976:PAD:360349.360355}
Levin, J. Z. (1976). A parametric algorithm for drawing pictures of solid objects composed of quadric surfaces. Communications of the ACM 19 (10), 555-563. http://doi.org/10.1145/360349.360355

\bibitem{Li2006a}
Li, Y. Bin. (2006). A new approach for constructing subresultants. Applied Mathematics and Computation 183 (1), 471--476. http://doi.org/10.1016/j.amc.2006.05.120

\bibitem{Miller1987}
Miller, J. R. (1987). Geometric approaches to nonplanar quadric surface intersection curves. ACM Transactions on Graphics 6 (4), 274--307. http://doi.org/10.1145/35039.35041

\bibitem{MILLER199555}
Miller, J. R., \& Goldman, R. N. (1995). Geometric Algorithms for Detecting and Calculating All Conic Sections in the Intersection of Any 2 Natural Quadric Surfaces. Graphical Models and Image Processing 57 (1), 55-66. http://doi.org/10.1006/gmip.1995.1006

\bibitem{Mourrain2005}
Mourrain, B., T\'ecourt, J. P., \& Teillaud, M. (2005). On the computation of an arrangement of quadrics in 3D. In Computational Geometry: Theory and Applications 30 (2), 145--164. http://doi.org/10.1016/j.comgeo.2004.05.003

\bibitem{Schomer2006}
Schomer, E., \& Wolpert, N. (2006). An exact and efficient approach for computing a cell in an arrangement of quadrics. Computational Geometry: Theory and Applications 33 (1--2), 65--97. http://doi.org/10.1016/j.comgeo.2004.02.007

\bibitem{Sendra1999}
Sendra, J. R., Winkler, F. (1999). Algorithms for rational real algebraic curves. Fundamenta Informaticae 39 (1--2), 211--228.

\bibitem{Shene:1994:LDI:195826.197316}
Shene, C.-K., \& Johnstone, J. K. (1994). On the lower degree intersections of two natural quadrics. ACM Transactions on Graphics 13 (4), 400--424. http://doi.org/10.1145/195826.197316

\bibitem{WANG2003401}
Wang, W., Goldman, R., \& Tu, C. (2003). Enhancing Levin's method for computing quadric--surface intersections. Computer Aided Geometric Design 20 (7), 401--422. http://doi.org/10.1016/S0167-8396(03)00081-5

\bibitem{Wang2002}
Wang, W., Joe, B., \& Goldman, R. (2002). Computing quadric surface intersections based on an analysis of plane cubic curves. Graphical Models 64 (6), 335--367. http://doi.org/10.1016/S1077-3169(02)00018-7

\bibitem{Wilf1993}
Wilf, I., \& Manor, Y. (1993). Quadric-surface intersection curves: shape and structure. Computer-Aided Design 25 (10), 633--643. http://doi.org/10.1016/0010-4485(93)90018-J

\end{thebibliography}

\section{Annex}

{\scriptsize
\begin{longtable}[c]{| c | c |}
 \hline
\textbf{Example} & \textbf{Quadrics} \\
 \hline
 \endfirsthead
 \hline
 \textbf{Example} & \textbf{Quadrics}\\
 \hline
 \endhead
\multirow{2}{*}{Example 1} & ${z}^{2}+ \left( 2\,x-y+4 \right) z+12\,{x}^{2}-24\,xy+11\,{y}^{2}+20\, x-18\,y+12$\\ 
                           &${z}^{2}+ \left( -{\frac {26\,x}{9}}+5/9\,y+{\frac{20}{9}} \right) z+4/
3\,{x}^{2}-1/9\,{y}^{2}-{\frac {20\,x}{9}}+2/9\,y+4/3$\\ \hline
\multirow{2}{*}{Example 2} & ${z}^{2}+ \left( -x+3 \right) z+{x}^{2}-3/2\,xy+1/2\,{y}^{2}-x-y/2+2$\\ 
                           &${z}^{2}-8\,{x}^{2}+18\,xy-4\,xz-7\,{y}^{2}-12\,x+10\,y-4$\\ \hline
\multirow{2}{*}{Example 3} & ${z}^{2}+ \left( -8/5\,x+{\frac{12}{5}} \right) z-{\frac {12\,{x}^{2}}{
5}}+{\frac {28\,xy}{5}}-{\frac {13\,{y}^{2}}{5}}-{\frac {32\,x}{5}}+{
\frac {24\,y}{5}}-4/5
$\\ 
                           &${z}^{2}+ \left( -{\frac {16\,x}{25}}+{\frac{84}{25}} \right) z-{\frac 
{28\,{x}^{2}}{25}}+{\frac {12\,xy}{5}}-{\frac {27\,{y}^{2}}{25}}-{
\frac {112\,x}{25}}+{\frac {48\,y}{25}}+{\frac{52}{25}}
$\\ \hline
\multirow{2}{*}{Example 4} & ${z}^{2}+ \left( -{\frac {88\,x}{43}}+{\frac{84}{43}} \right) z+{\frac 
{60\,{x}^{2}}{43}}-{\frac {28\,xy}{43}}+{\frac {13\,{y}^{2}}{43}}-{
\frac {64\,x}{43}}-{\frac {24\,y}{43}}+{\frac{52}{43}}$\\ 
                           & ${z}^{2}+ \left( -{\frac {80\,x}{23}}+{\frac{12}{23}} \right) z+{\frac 
{76\,{x}^{2}}{23}}-{\frac {60\,xy}{23}}+{\frac {27\,{y}^{2}}{23}}+{
\frac {16\,x}{23}}-{\frac {48\,y}{23}}-{\frac{4}{23}}$\\ \hline
\multirow{2}{*}{Example 5} & ${z}^{2}+ \left( -2\,x+2 \right) z+3/2\,xy-1/2\,{y}^{2}-3\,x+y/2+1$\\ 
                           & ${z}^{2}+ \left( -{\frac {16\,x}{7}}+{\frac{12}{7}} \right) z+{\frac {
10\,xy}{7}}-3/7\,{y}^{2}-{\frac {20\,x}{7}}+2/7\,y+4/7
$\\ \hline
\multirow{2}{*}{Example 6} & ${z}^{2}+ \left( -2\,x+2 \right) z+3/2\,xy-1/2\,{y}^{2}-3\,x+y/2+1$\\ 
                           & ${z}^{2}+ \left( -{\frac {32\,x}{15}}+{\frac{28}{15}} \right) z+{\frac 
{8\,{x}^{2}}{15}}+2/3\,xy-1/5\,{y}^{2}-{\frac {12\,x}{5}}+2/15\,y+4/5$\\ \hline
\multirow{2}{*}{Example 7} & ${z}^{2}+ \left( -2\,x+2 \right) z+3/2\,xy-1/2\,{y}^{2}-3\,x+y/2+1$\\ 
                           & ${z}^{2}+ \left( -{\frac {16\,x}{9}}+{\frac{20}{9}} \right) z+{\frac {
10\,xy}{9}}-1/3\,{y}^{2}-{\frac {20\,x}{9}}+2/9\,y+4/3$\\ \hline
\multirow{2}{*}{Example 8} & ${z}^{2}+ \left( -3\,x+2\,y+1 \right) z+4\,{x}^{2}-4\,xy+1/2\,{y}^{2}-2
\,x+2\,y$\\ 
                           & ${z}^{2}+ \left( -{\frac {30\,x}{13}}+{\frac {4\,y}{13}}+{\frac{22}{13}
} \right) z+{\frac {16\,{x}^{2}}{13}}-{\frac {4\,xy}{13}}-{\frac {28\,
x}{13}}+{\frac {4\,y}{13}}+{\frac{8}{13}}
$\\ \hline
\multirow{2}{*}{Example 9} & ${z}^{2}+ \left( -8/3\,x+4/3 \right) z+8/3\,{x}^{2}-2\,xy+2/3\,{y}^{2}-
4/3\,x-2/3\,y
$\\ 
                           & ${z}^{2}+ \left( -{\frac {12\,x}{7}}+{\frac{16}{7}} \right) z+{\frac {
16\,{x}^{2}}{7}}-2\,xy+4/7\,{y}^{2}-{\frac {12\,x}{7}}-2/7\,y+{\frac{8
}{7}}
$\\ \hline
\multirow{2}{*}{Example 10} & ${z}^{2}+ \left( -x+3 \right) z-2\,{x}^{2}+3\,xy-{y}^{2}-4\,x+y+2$\\ 
                           & ${z}^{2}+ \left( -2\,x+2 \right) z+{\frac {4\,{x}^{2}}{13}}+{\frac {10
\,xy}{13}}-3/13\,{y}^{2}-{\frac {32\,x}{13}}+2/13\,y+{\frac{12}{13}}$\\ \hline
\multirow{2}{*}{Example 11} & ${z}^{2}+ \left( -2\,x+2 \right) z+3/2\,xy-1/2\,{y}^{2}-3\,x+y/2+1$\\ 
                           & ${z}^{2}+4\,z+8\,{x}^{2}-10\,xy+3\,{y}^{2}+4\,x-2\,y+4$\\ \hline
\multirow{2}{*}{Example 12} & ${z}^{2}+ \left( -4\,x+2\,y+2 \right) z+2\,{x}^{2}-xy-6\,x+3\,y$\\ 
                           & ${z}^{2}+ \left( -2\,x-y+6 \right) z-2\,{x}^{2}+3\,xy-10\,x+y+4$\\ \hline
\multirow{2}{*}{Example 13} & ${z}^{2}+ \left( -2\,x+2 \right) z+3\,{x}^{2}-3\,xy+5/4\,{y}^{2}-2\,y+2$\\ 
                           & ${z}^{2}+ \left( -{\frac {12\,x}{5}}+8/5 \right) z+2\,{x}^{2}-6/5\,xy+1
/2\,{y}^{2}-8/5\,x-4/5\,y+4/5
$\\ \hline
\multirow{2}{*}{Example 14} & ${z}^{2}+ \left( -2\,x+2 \right) z+{x}^{2}-xy+3/4\,{y}^{2}-2\,y+2$\\ 
                           & ${z}^{2}+ \left( -4/3\,x+8/3 \right) z+2/3\,{x}^{2}+2/3\,xy-1/2\,{y}^{2
}-8/3\,x+4/3\,y+4/3$\\ \hline
\multirow{2}{*}{Example 15} & ${z}^{2}+ \left( -2\,x+2 \right) z+1/2\,{y}^{2}-2\,y+2$\\ 
                           & ${z}^{2}+ \left( -4/3\,x+8/3 \right) z+4/3\,{x}^{2}-1/3\,{y}^{2}-8/3\,x
+4/3\,y+4/3$\\ \hline
\multirow{2}{*}{Example 16} & ${z}^{2}+ \left( -8/3\,x+4/3 \right) z-4/3\,{x}^{2}+4\,xy-5/3\,{y}^{2}-
16/3\,x+8/3\,y-4/3$\\ 
                           & ${z}^{2}+ \left( -{\frac {16\,x}{11}}+{\frac{28}{11}} \right) z-4/11\,{
x}^{2}+{\frac {12\,xy}{11}}-{\frac {5\,{y}^{2}}{11}}-{\frac {32\,x}{11
}}+{\frac {8\,y}{11}}+{\frac{12}{11}}$\\ \hline
\multirow{2}{*}{Example 17} & ${z}^{2}+ \left( -{\frac {32\,x}{15}}+{\frac{28}{15}} \right) z+{\frac 
{8\,{x}^{2}}{15}}+4/5\,xy-1/3\,{y}^{2}-8/3\,x+{\frac {8\,y}{15}}+{
\frac{8}{15}}$\\ 
                           & ${z}^{2}+ \left( -{\frac {16\,x}{13}}-{\frac {32\,y}{39}}+{\frac{36}{13
}} \right) z+{\frac {8\,{x}^{2}}{39}}+{\frac {28\,xy}{39}}+1/13\,{y}^{
2}-{\frac {24\,x}{13}}-{\frac {40\,y}{39}}+{\frac{24}{13}}$\\ \hline
\multirow{2}{*}{Example 18} & ${z}^{2}+ \left( -8/3\,x+4/3 \right) z-4/3\,{x}^{2}+4\,xy-5/3\,{y}^{2}-
16/3\,x+8/3\,y-4/3$\\ 
                           & ${z}^{2}+ \left( -{\frac {24\,x}{7}}+{\frac {32\,y}{21}}+4/7 \right) z+
{\frac {52\,{x}^{2}}{21}}-4/3\,xy-1/7\,{y}^{2}-{\frac {16\,x}{7}}+{
\frac {40\,y}{21}}-4/7$\\ \hline
\multirow{2}{*}{Example 19} & ${z}^{2}+ \left( -8/3\,x+4/3 \right) z+4/3\,xy-1/3\,{y}^{2}-8/3\,x$\\ 
                           & ${z}^{2}+ \left( -8/5\,x+{\frac{12}{5}} \right) z+{\frac {12\,xy}{5}}-7
/5\,{y}^{2}-{\frac {24\,x}{5}}+{\frac {16\,y}{5}}$\\ \hline
\multirow{2}{*}{Example 20} & ${z}^{2}+ \left( -{\frac {24\,x}{11}}+{\frac{20}{11}} \right) z+{\frac 
{8\,{x}^{2}}{11}}+4/11\,xy-1/11\,{y}^{2}-{\frac {24\,x}{11}}+{\frac{8}
{11}}$\\ 
                           & ${z}^{2}+ \left( -8/3\,x+4/3 \right) z+8/3\,{x}^{2}-4\,xy+7/3\,{y}^{2}+
8/3\,x-16/3\,y+8/3$\\ \hline
\multirow{2}{*}{Example 21} & ${z}^{2}+ \left( -{\frac {24\,x}{11}}+{\frac{20}{11}} \right) z+{\frac 
{8\,{x}^{2}}{11}}+4/11\,xy-1/11\,{y}^{2}-{\frac {24\,x}{11}}+{\frac{8}
{11}}$\\ 
                           & ${z}^{2}+ \left( -{\frac {504\,x}{253}}+{\frac{508}{253}} \right) z+{
\frac {248\,{x}^{2}}{253}}+{\frac {12\,xy}{253}}-{\frac {7\,{y}^{2}}{
253}}-{\frac {520\,x}{253}}+{\frac {16\,y}{253}}+{\frac{248}{253}}$\\ \hline
\multirow{2}{*}{Example 22} & ${z}^{2}-4\,xz+2\,xy-1/2\,{y}^{2}-4\,x-2$\\ 
                           & ${z}^{2}-4\,xz-2\,xy+5/2\,{y}^{2}+4\,x-8\,y+2$\\ \hline
\multirow{2}{*}{Example 23} & ${z}^{2}+ \left( -{\frac {24\,x}{11}}+{\frac{20}{11}} \right) z+{\frac 
{8\,{x}^{2}}{11}}+4/11\,xy-1/11\,{y}^{2}-{\frac {24\,x}{11}}+{\frac{8}
{11}}$\\ 
                           & ${z}^{2}+ \left( -{\frac {24\,x}{11}}+{\frac{20}{11}} \right) z+{\frac 
{8\,{x}^{2}}{11}}-4/11\,xy+{\frac {5\,{y}^{2}}{11}}-{\frac {8\,x}{11}}
-{\frac {16\,y}{11}}+{\frac{16}{11}}$\\ \hline
\multirow{2}{*}{Example 24} & ${z}^{2}+ \left( -{\frac {24\,x}{11}}+{\frac{20}{11}} \right) z+{\frac 
{8\,{x}^{2}}{11}}+4/11\,xy-1/11\,{y}^{2}-{\frac {24\,x}{11}}+{\frac{8}
{11}}$\\ 
                           & ${z}^{2}+ \left( -{\frac {24\,x}{11}}+{\frac{20}{11}} \right) z-{\frac 
{248\,{x}^{2}}{11}}+{\frac {252\,xy}{11}}-{\frac {59\,{y}^{2}}{11}}-{
\frac {8\,x}{11}}-{\frac {16\,y}{11}}+{\frac{16}{11}}
$\\ \hline
\multirow{2}{*}{Example 25} & ${z}^{2}+ \left( -x+3 \right) z+2\,{x}^{2}-{y}^{2}-6\,x+4\,y$\\ 
                           & ${z}^{2}+ \left( -6\,x-2 \right) z+2\,{y}^{2}+4\,x-8\,y$\\\hline
\multirow{2}{*}{Example 26} & ${z}^{2}+ \left( -{\frac {88\,x}{43}}+{\frac{84}{43}} \right) z+{\frac 
{52\,{x}^{2}}{43}}-{\frac {4\,xy}{43}}-1/43\,{y}^{2}-{\frac {96\,x}{43
}}+{\frac {8\,y}{43}}+{\frac{36}{43}}$\\ 
                           & ${z}^{2}+ \left( -{\frac {32\,x}{17}}+{\frac{36}{17}} \right) z+{\frac 
{52\,{x}^{2}}{51}}-{\frac {4\,xy}{51}}-{\frac {{y}^{2}}{51}}-{\frac {
112\,x}{51}}+{\frac {8\,y}{51}}+{\frac{52}{51}}$\\\hline
\multirow{2}{*}{Example 27} & ${z}^{2}+ \left( -{\frac {88\,x}{43}}+{\frac{84}{43}} \right) z+{\frac 
{52\,{x}^{2}}{43}}-{\frac {4\,xy}{43}}-1/43\,{y}^{2}-{\frac {96\,x}{43
}}+{\frac {8\,y}{43}}+{\frac{36}{43}}$\\ 
                           & ${z}^{2}+ \left( -{\frac {24\,x}{5}}-4/5 \right) z+{\frac {28\,{x}^{2}
}{5}}-4/5\,xy-1/5\,{y}^{2}-{\frac {16\,x}{5}}+8/5\,y-{\frac{12}{5}}$\\ \hline
\multirow{2}{*}{Example 28} & ${z}^{2}+4\,z-4\,{x}^{2}+6\,xy-2\,{y}^{2}-4\,x+2\,y+4$\\ 
                           & ${z}^{2}+ \left( -8/5\,x+{\frac{12}{5}} \right) z-4/5\,{x}^{2}+{\frac {
14\,xy}{5}}-6/5\,{y}^{2}-4\,x+2\,y+4/5$\\ \hline
\multirow{2}{*}{Example 29} & ${z}^{2}+4\,z-4\,{x}^{2}+6\,xy-2\,{y}^{2}-4\,x+2\,y+4$\\ 
                           & ${z}^{2}+ \left( -{\frac {16\,x}{9}}+{\frac{20}{9}} \right) z+{\frac {
14\,xy}{9}}-2/3\,{y}^{2}-{\frac {28\,x}{9}}+{\frac {10\,y}{9}}+{\frac{
8}{9}}$\\ \hline
\multirow{2}{*}{Example 30} & ${z}^{2}+4\,z+8\,{x}^{2}-12\,xy+5\,{y}^{2}+8\,x-8\,y+8$\\ 
                           & ${z}^{2}+ \left( -{\frac {16\,x}{7}}+{\frac{12}{7}} \right) z+{\frac {
12\,xy}{7}}-5/7\,{y}^{2}-{\frac {24\,x}{7}}+{\frac {8\,y}{7}}$\\ \hline
\multirow{2}{*}{Example 31} & ${z}^{2}+ \left( -8/3\,x+4/3 \right) z+4/3\,xy-1/3\,{y}^{2}-8/3\,x$\\ 
                           & ${z}^{2}+ \left( -8/3\,x+4/3 \right) z+8/3\,{x}^{2}-4\,xy+7/3\,{y}^{2}+
8/3\,x-16/3\,y+8/3$\\ \hline
\multirow{2}{*}{Example 32} & ${z}^{2}+ \left( -{\frac {24\,x}{11}}+{\frac{20}{11}} \right) z+{\frac 
{8\,{x}^{2}}{11}}+4/11\,xy-1/11\,{y}^{2}-{\frac {24\,x}{11}}+{\frac{8}
{11}}$\\ 
                           & ${z}^{2}+ \left( -{\frac {24\,x}{11}}+{\frac{20}{11}} \right) z+{\frac 
{16\,{x}^{2}}{11}}-{\frac {12\,xy}{11}}+{\frac {7\,{y}^{2}}{11}}-{
\frac {8\,x}{11}}-{\frac {16\,y}{11}}+{\frac{16}{11}}
$\\ \hline
\multirow{2}{*}{Example 33} & ${z}^{2}+ \left( -8/3\,x+4/3 \right) z+4/3\,xy-1/3\,{y}^{2}-8/3\,x$\\ 
                           & ${z}^{2}+ \left( -8\,x+16/3\,y-4 \right) z+16/3\,{x}^{2}-4\,xy-1/3\,{y}
^{2}-8/3\,x+16/3\,y-16/3$\\ \hline
\multirow{2}{*}{Example 34} & ${z}^{2}+ \left( -{\frac {24\,x}{11}}+{\frac{20}{11}} \right) z+{\frac 
{8\,{x}^{2}}{11}}+4/11\,xy-1/11\,{y}^{2}-{\frac {24\,x}{11}}+{\frac{8}
{11}}$\\ 
                           & ${z}^{2}+ \left( -{\frac {40\,x}{11}}+{\frac {16\,y}{11}}+4/11 \right) 
z+{\frac {24\,{x}^{2}}{11}}-{\frac {12\,xy}{11}}-1/11\,{y}^{2}-{\frac 
{24\,x}{11}}+{\frac {16\,y}{11}}-{\frac{8}{11}}
$\\ \hline
\multirow{2}{*}{Example 35} & ${z}^{2}+ \left( -8/3\,x+4/3 \right) z+4/3\,xy-1/3\,{y}^{2}-8/3\,x$\\ 
                           & ${z}^{2}+ \left( -8\,x+16/3\,y-4 \right) z+16/3\,{x}^{2}-4\,xy-1/3\,{y}
^{2}-8/3\,x+16/3\,y-16/3$\\ \hline
\multirow{2}{*}{Example 36} & ${z}^{2}+ \left( -{\frac {24\,x}{11}}+{\frac{20}{11}} \right) z+4/11\,{
x}^{2}+{\frac {8\,xy}{11}}-2/11\,{y}^{2}-{\frac {24\,x}{11}}+{\frac{8}
{11}}$\\ 
                           & ${z}^{2}+ \left( -{\frac {24\,x}{11}}+{\frac{20}{11}} \right) z-{\frac 
{12\,{x}^{2}}{11}}+{\frac {32\,xy}{11}}-{\frac {10\,{y}^{2}}{11}}-{
\frac {40\,x}{11}}+{\frac {8\,y}{11}}+{\frac{8}{11}}$\\ \hline
\multirow{2}{*}{Example 37} & ${z}^{2}+ \left( -5/2\,x+y/2+3/2 \right) z+2\,{x}^{2}-xy+1/4\,{y}^{2}-3
\,x+y$\\ 
                           & ${z}^{2}+ \left( -7/2\,x+y+3/2 \right) z+2\,{x}^{2}-xy+1/4\,{y}^{2}-5\,
x+2\,y$\\ \hline
\multirow{2}{*}{Example 38} &${z}^{2}+ \left( -3/2\,x-y/2+5/2 \right) z+xy-1/4\,{y}^{2}-x-y+2$\\ 
                           & ${z}^{2}+ \left( -x/2-y+5/2 \right) z+xy-1/4\,{y}^{2}+x-2\,y+2$\\ \hline
\multirow{2}{*}{Example 39} &${z}^{2}+ \left( -7/4\,x-y/4+9/4 \right) z+1/2\,{x}^{2}+1/2\,xy-1/8\,{y
}^{2}-3/2\,x-y/2+3/2$\\ 
                           & ${z}^{2}+ \left( -5/4\,x-y/2+9/4 \right) z+1/2\,{x}^{2}+1/2\,xy-1/8\,{y
}^{2}-x/2-y+3/2$\\ \hline
\multirow{2}{*}{Example 40} &${z}^{2}+ \left( -8/5\,x+{\frac{12}{5}} \right) z-4/5\,{x}^{2}+{\frac {
12\,xy}{5}}-4/5\,{y}^{2}-{\frac {16\,x}{5}}+4/5\,y+8/5$\\ 
                           & ${z}^{2}+ \left( -{\frac {16\,x}{13}}+{\frac{36}{13}} \right) z-{\frac 
{20\,{x}^{2}}{13}}+{\frac {36\,xy}{13}}-{\frac {14\,{y}^{2}}{13}}-{
\frac {48\,x}{13}}+{\frac {20\,y}{13}}+{\frac{16}{13}}$\\ \hline
\multirow{2}{*}{Example 41} &${z}^{2}+ \left( -{\frac {32\,x}{17}}+{\frac{36}{17}} \right) z+{\frac 
{32\,xy}{17}}-{\frac {13\,{y}^{2}}{17}}-{\frac {64\,x}{17}}+{\frac {20
\,y}{17}}+{\frac{16}{17}}
$\\ 
                           & ${z}^{2}+ \left( -8/3\,x+4/3 \right) z+{\frac {32\,{x}^{2}}{9}}-{\frac 
{8\,xy}{9}}+1/9\,{y}^{2}-{\frac {32\,x}{9}}+4/9\,y+{\frac{8}{9}}$\\ \hline
\multirow{2}{*}{Example 42} &${z}^{2}+ \left( -{\frac {24\,x}{13}}+{\frac{28}{13}} \right) z+{\frac 
{24\,xy}{13}}-{\frac {10\,{y}^{2}}{13}}-{\frac {48\,x}{13}}+{\frac {16
\,y}{13}}+{\frac{12}{13}}$\\ 
                           & ${z}^{2}+ \left( -{\frac {16\,x}{5}}+4/5 \right) z+{\frac {24\,{x}^{2}
}{5}}-8/5\,xy+2/5\,{y}^{2}-{\frac {16\,x}{5}}+4/5$\\ \hline
\multirow{2}{*}{Example 43} &${z}^{2}+ \left( -{\frac {128\,x}{63}}+{\frac{124}{63}} \right) z+{
\frac {176\,{x}^{2}}{63}}-{\frac {8\,xy}{21}}-{\frac {13\,{y}^{2}}{63}
}-{\frac {304\,x}{63}}+{\frac {76\,y}{63}}+{\frac{8}{7}}$\\ 
                           & ${z}^{2}+ \left( -{\frac {136\,x}{71}}+{\frac{148}{71}} \right) z+{
\frac {368\,{x}^{2}}{71}}-{\frac {208\,xy}{71}}+{\frac {29\,{y}^{2}}{
71}}-{\frac {336\,x}{71}}+{\frac {92\,y}{71}}+{\frac{80}{71}}$\\ \hline
\multirow{2}{*}{Example 44} &${z}^{2}+ \left( -{\frac {16\,x}{7}}+{\frac{12}{7}} \right) z+{\frac {
16\,xy}{7}}-5/7\,{y}^{2}-{\frac {32\,x}{7}}+4/7\,y+{\frac{8}{7}}$\\ 
                           & ${z}^{2}+ \left( -8/5\,x+{\frac{12}{5}} \right) z+{\frac {16\,{x}^{2}}{
15}}+{\frac {8\,xy}{15}}-{\frac {7\,{y}^{2}}{15}}-{\frac {64\,x}{15}}+
4/3\,y+{\frac{16}{15}}$\\ \hline
\multirow{2}{*}{Example 45} &${z}^{2}+ \left( -{\frac {24\,x}{11}}+{\frac{20}{11}} \right) z+{\frac 
{24\,xy}{11}}-{\frac {48\,x}{11}}-{\frac {8\,{y}^{2}}{11}}+{\frac {8\,
y}{11}}+{\frac{12}{11}}$\\ 
                           & ${z}^{2}+ \left( -{\frac {32\,x}{19}}+{\frac{44}{19}} \right) z+{\frac 
{24\,{x}^{2}}{19}}+{\frac {8\,xy}{19}}-{\frac {80\,x}{19}}-{\frac {8\,
{y}^{2}}{19}}+{\frac {24\,y}{19}}+{\frac{20}{19}}$\\ \hline
\multirow{2}{*}{Example 46} &${z}^{2}+ \left( -{\frac {128\,x}{65}}+{\frac{132}{65}} \right) z-{
\frac {64\,{x}^{2}}{65}}+{\frac {112\,xy}{65}}-{\frac {29\,{y}^{2}}{65
}}-{\frac {96\,x}{65}}+{\frac {4\,y}{65}}+{\frac{48}{65}}$\\ 
                           & ${z}^{2}+ \left( -{\frac {40\,x}{19}}+{\frac{36}{19}} \right) z-{\frac 
{128\,{x}^{2}}{57}}+{\frac {56\,xy}{19}}-{\frac {13\,{y}^{2}}{19}}-{
\frac {64\,x}{57}}-{\frac {4\,y}{19}}+{\frac{40}{57}}$\\ \hline
\multirow{2}{*}{Example 47} &${z}^{2}+ \left( -{\frac {32\,x}{17}}+{\frac{36}{17}} \right) z-{\frac 
{16\,{x}^{2}}{17}}+{\frac {24\,xy}{17}}-{\frac {7\,{y}^{2}}{17}}-{
\frac {16\,x}{17}}+{\frac {4\,y}{17}}+{\frac{8}{17}}$\\ 
                           & ${z}^{2}+ \left( -8/3\,x+4/3 \right) z-16/3\,{x}^{2}+16/3\,xy-{y}^{2}+{
\frac {16\,x}{9}}-4/3\,y$\\ \hline
\multirow{2}{*}{Example 48} &${z}^{2}+ \left( -{\frac {32\,x}{17}}+{\frac{36}{17}} \right) z-{\frac 
{16\,{x}^{2}}{17}}+{\frac {24\,xy}{17}}-{\frac {7\,{y}^{2}}{17}}-{
\frac {16\,x}{17}}+{\frac {4\,y}{17}}+{\frac{8}{17}}$\\ 
                           & ${z}^{2}+ \left( -8/3\,x+4/3 \right) z-16/3\,{x}^{2}+16/3\,xy-{y}^{2}+{
\frac {16\,x}{9}}-4/3\,y$\\ \hline
\multirow{2}{*}{Example 49} &${z}^{2}+ \left( -{\frac {128\,x}{65}}+{\frac{132}{65}} \right) z-{
\frac {56\,{x}^{2}}{13}}+{\frac {288\,xy}{65}}-{\frac {64\,{y}^{2}}{65
}}-{\frac {16\,x}{65}}-{\frac {32\,y}{65}}+{\frac{44}{65}}$\\ 
                           & ${z}^{2}+ \left( -{\frac {40\,x}{17}}+{\frac{28}{17}} \right) z-{\frac 
{16\,{x}^{2}}{17}}+{\frac {32\,xy}{17}}-{\frac {7\,{y}^{2}}{17}}-{
\frac {16\,x}{17}}-{\frac {4\,y}{17}}+{\frac{8}{17}}$\\ \hline
\multirow{2}{*}{Example 50} &${z}^{2}+ \left( -2/3\,x+2/3\,y \right) z+1/3\,{x}^{2}+1/3\,{y}^{2}-1/3$\\ 
                           & ${z}^{2}+ \left( -2/17\,x+{\frac {24\,y}{17}}-2/17 \right) z+1/17\,{x}^
{2}+2/17\,x-{\frac{3}{17}}+{\frac {12\,{y}^{2}}{17}}$\\ \hline
\end{longtable}}

\end{document}